\documentclass[12pt]{article}

\usepackage[a4paper]{geometry}
\usepackage{latexsym}
\usepackage[pdftex]{graphicx}
\usepackage[pdftex]{color}

\usepackage{epsfig,longtable}
\usepackage{makeidx}
\usepackage{varioref}
\usepackage{xr}
\usepackage{amsmath}
\usepackage{amsthm}

\usepackage{amssymb}
\usepackage{url}
\usepackage[authoryear]{natbib}
\usepackage{setspace}
\usepackage{subfigure}
\usepackage{multirow}
\usepackage{dsfont}

\newtheorem{procedure}{Procedure}[section]
\newtheorem{theorem}{Theorem}[section]

\newtheorem{lemma}{Lemma}[section]
\newtheorem{corollary}{Corollary}[section]
\newtheorem{remark}{Remark}[section]
\newtheorem{example}{Example}[section]
\newtheorem{cor}{Corollary}[section]
\numberwithin{equation}{section}

\begin{document}

\noindent {\sffamily\bfseries\Large  Post-selection estimation and testing following aggregated association tests }

\noindent%
\textsf{Ruth Heller}, Department of
Statistics and Operations Research, Tel-Aviv university, Tel-Aviv  6997801, 
Israel,  \textsf{E-mail:} ruheller@gmail.com\\

\textsf{Amit Meir},
Department of statistics, University of Washington, Seattle, WA, \textsf{E-mail:}
amitmeir@uw.edu.\\

\textsf{Nilanjan Chatterjee}, Department of Biostatistics, Bloomberg School of Public Health, and Department of Oncology, School of Medicine, Johns Hopkins University, Baltimore, MD 21205, U.S.A.,\textsf{E-mail:} nchatte2@jhu.edu

\begin{abstract}
The practice of pooling several individual test statistics to form aggregate tests is common in many statistical application where individual tests may be underpowered.  
While selection by aggregate tests can serve to increase power, the selection process invalidates the individual test-statistics, making it difficult to identify the ones that drive the signal in follow-up inference. Here, we develop a general approach for valid inference following selection by aggregate testing.  We  present novel powerful post-selection tests for the individual null hypotheses which are exact for the normal model and asymptotically justified otherwise. Our approach relies on the ability to characterize the distribution of the individual test statistics after conditioning on the event of selection. We provide efficient algorithms for estimation of the post-selection maximum-likelihood estimates and suggest confidence intervals which rely on a novel switching regime for good coverage guarantees. We validate our methods via  comprehensive simulation studies and apply them to data from the Dallas Heart Study, demonstrating that single variant association discovery following selection by an aggregated test is indeed possible in practice.
\end{abstract}

% \keywords{selective inference; aggregate tests; conditional $p$-value;  conditional confidence interval;   multiple testing}

\section{Introduction}

Many modern scientific investigations involve simultaneous testing of many thousands of hypotheses. Valid testing of large number of hypotheses requires strict multiple-testing adjustments, making it difficult to identify signals in the data if the signal is weak or sparse. One possible remedy is to pool groups of related test statistics into aggregate tests. This practice reduces the amount of multiplicity correction that needs to be applied and may assist in identifying weak signals that are spread over a number of test statistics. However, once an `interesting' group of hypotheses has been identified, it may also be of interest to perform inference within the group in order to identify the individual test statistics that drive the signal. 

In many scientific fields, there exist a natural predefined grouping of features of interest. In neuroscience, functional magnetic resonance imaging (fMRI) studies aim to identify the locations of activation while a subject is involved in a cognitive task. The individual null hypotheses of no activation are at the voxel level, and regions of interest can be tested for activation by aggregating the measured signals at the voxel level \citep{Penny03, Benjamini07}. Following identification of the regions of interest, it is meaningful to localize the signal within the region.  In microbiome research, the operational taxonomic units (OTUs) are grouped into taxonomic classifications such as species level and genus level. The data for the individual null hypotheses of no association between OTU and phenotype can be aggregated in order to test the null hypotheses of no association at the species or genus level \citep{Bogomolov17}. Here as well, following identification of the family associated with the phenotype, it is of interest to identify the OTUs within the family that drive the association. In genome-wide association studies (GWAS), the disease being analyzed may have multiple subtypes of interest. The standard analysis aim is to identify the SNPs associated with the overall disease, but another important aim is to identify associations with specific sub-types of the disease \citep{Bhattacharjee12}. In genetic association studies, there is also a natural grouping of the genome, since genes are comprised of single variants. The test statistics of single variants within a gene can be aggregated into a test statistic for powerful identification of associations at the gene level \citep{Wu11, Bhattacharjee12, Derkach14, Yoo16}. Following identification at the gene level, it may be of interest to identify the single variants within the gene that drive the association.  

For a single group of features, let $\hat {\vec \beta}= (\hat \beta_1,\ldots, \hat \beta_m)'$ be the estimator for the vector of parameters of interest in the group, $\vec \beta$. 
Much research has focused on developing powerful aggregate tests for selecting the groups of interest, i.e., for testing at the group level the null hypothesis that $\vec \beta=\vec 0$. When $\hat {\vec \beta}$ has (approximately) a known covariance  and a normal distribution, classical test-statistics are the score, Wald, and likelihood ratio statistics, all of which have an asymptotic $\chi^2_m$ distribution.   A recent example is the work by \cite{Reid16}, which suggested novel tests that improve on classical tests. Other examples come from the field of statistical genetics, where many gene level tests have been recently proposed based on weighted linear or quadratic combinations of score statistics for analyzing genomic studies of  rare variants, see \cite{Derkach14} for a review. 

In this work we seek to develop methods for conducting inference on the coordinates of $\vec \beta$ following selection by an aggregate test. Failure to account for data driven selection of any kind can lead to biased inference. For example, in linear regression, if the relationship of the predictors with a response is assumed linear for a group, then selection by an aggregate test constraints the response vector to values for which the aggregate test $p$-value is below a threshold, and the post-selection distribution of the data is no longer multivariate normal but a truncated multivariate normal. Generally speaking, ignoring the selection will result in biased inference if there is dependence between the selection event and the individual test-statistic: if the individual test-statistic contributes to the selection by the aggregate test, then conditional on being selected the distribution of the individual test statistic is changed. 

Inference following selection is an emerging field of research, which is of great interest both in the statistics community and in application fields. In the multiple testing literature, \cite{Benjamini14} presented a  novel approach for  the problem of inference within families of hypotheses following selection of the families by any selection rule.  Marginal confidence intervals following selection are addressed in generality in \cite{Benjamini05},  from a  Bayesian perspective in \cite{Yekutieli12}, and for a specific selection rule (i.e., that the test statistic is larger than a certain threshold) in \cite{Weinstein13}. Significant progress has also been made in the regression context, where variables are first selected into the model, and inference on the selected variables follows. Failing to account for the data-driven variable selection process invalidates the inference \citep{Potscher91, Berk13, Fithian15}. Recent valid post model-selection procedures can be found in \cite{Berk13, Fithian15, Lee14, Lee16, Taylor16, Tian16, Meir17}.

Recently, \cite{Heller16} addressed the problem of identifying the individual studies with association with a feature, following selection of potential features by a meta-analysis of multiple independent studies.  We generalize the work of \cite{Heller16} to allow for (approximately) known dependence across the individual test statistics. In particular, this allows for valid testing of predictors in a generalized linear model that was selected via an aggregated test . We further develop methods for obtaining post-selection point estimates and confidence intervals.  We discuss the computation of maximum-likelihood estimates following aggregate testing and show that in the special case of aggregate testing with a Wald test, the problem of computing the multivariate conditional maximum likelihood estimator can be cast as a simple line-search problem. We also  discuss computation of post-selection confidence intervals which are based on inversion of the post-selection tests. Finally, we develop regime switching post-selection tests and confidence intervals that adapt to the unknown underlying sparsity of the signal, and thus have good power when the signal is sparse as well as when it is non-sparse. 

The paper is organized as follows. In \S~\ref{sec-model} we formally introduce our inference framework and goals. We develop theory for post-selection testing and estimation in \S~\ref{sec-testing} and \S~\ref{sec-estimation}, respectively.  We conduct empirical evaluation of our test-statistics and post-selection estimates in \S~\ref{sec-sim}. In \S~\ref{sec-motivating-example}, we apply our methods to a genomic application. Finally, \S~\ref{sec-discuss} concludes.

\section{The set-up and the inferential goals} \label{sec-model}
Let $\hat{\vec\beta} \sim N(\vec\beta, {\bf\Sigma})$ with $\bf\Sigma$ known, and suppose that we are interesting in performing inference on $\vec\beta\in \mathcal{R}^{m}$ if and only if we can reject an aggregate test for the global-null hypothesis that $\vec\beta = \vec 0$. For testing the global-null hypothesis, we use a quadratic test of the form $S = \hat{\vec\beta}'K\hat{\vec\beta} > S_{1-t_1} $ where $K$ is a semi positive-definite matrix and $S_{1-t_1}$ is the $1-t_1$ quantile of $S$ under the null-hypothesis. Setting ${\bf K} = {\bf \Sigma}^{(-1)}$ results in the well known Wald test statistic. The developments when group selection is by a linear aggregate test $S = a'\hat{\vec\beta}$ are similar and detailed in Appendix \ref{sec-linear-aggregate-test}. 

The value of $t_1$ comes from the analysis at the group level. For example, in genomics, when the group is the gene, then typically $t_1 \approx \alpha/20000$. This is because the Bonferroni procedure is commonly used for identifying genes associated with phenotypes using aggregate tests, so the FWER on the family of $\sim 20,000$ genes is controlled at level $\alpha$.

Given that an aggregate test has been rejected at a level $t_1$, our aim is to infer on the  parameters $\beta_1,\ldots, \beta_m$. For $j\in \{1,\ldots,m\}$, let  
$$H_{j}: \beta_j=0.$$ 
Our first aim is to test the family of hypotheses $\{H_j: j=1,\ldots,m \}$ if $p^G\leq t_1$, with FWER or FDR control. The conditional FWER and FDR (introduced in \citealp{Heller16}) for the selected group are, respectively, 
  $E(I[V>0]|S>S_{1-t_1})$ and 
  $E(V/\max\{R,1 \}|S>S_{1-t_1}),$
where $V$ and $R$ are the number of false and total rejections in the group. We provide procedures for conditional FWER/FDR control in \S~\ref{sec-testing}. 

Our second aim is to estimate the magnitude of the regression coefficients   $\beta_1,\ldots, \beta_m$ given selection. Denoting the likelihood for $\vec \beta$ by   $\mathcal L(\vec \beta)$, the conditional likelihood can be written as:
$$
\mathcal{L}(\vec \beta| S > S_{1-t_1}) = 
\frac{\mathcal L(\vec \beta)}{P_{\vec \beta}( S > S_{1-t_1})} I\{S > S_{1-t_1}\}.
$$
We propose to use the maximizer of the conditional likelihood as a  point estimate, and we show how to obtain it, as well as confidence intervals, in \S~\ref{sec-estimation}.  

\begin{example} \label{rem-glm} Generalized Linear Models.\emph{ 
Suppose we observe a response vector $\vec y = (y_1,\ldots, y_n)'\in \mathcal R^n$, and  $m$ predictors of interest in a group (e.g., the single variants in a gene),  $\vec X_j, j=1, \ldots,m$. Let $\vec V_j, j=1,\ldots, k$ be a set of additional covariates to be accounted for in the model (e.g., environmental factors or ancestry variables in GWAS).
Suppose that we are interested in modelling the relationship of the predictors in a group with the response vector using a generalized linear model. So, we assume that $y_i \sim f_{\theta_i}$, an exponential family distribution with canonical parameter $\theta = g^{-1}(\eta_i) \in \Theta$ for some continuous link function $g:\Theta\rightarrow \mathcal{R}$ and 
\begin{equation}\label{eq-linear-model-2}
\vec \eta = \alpha_0 + \sum_{l=1}^k {\vec V}_l \alpha_l+  \sum_{j=1}^m {\vec X}_j \beta_j.%+\vec \epsilon
\end{equation}
In the case of linear regression, $g(\eta_i) = \eta_i$ is the identity function and $y_i \sim N(\eta_i, \sigma^{2})$. If $\vec X_1,\ldots,\vec X_m$  explain little of the variance in $\vec y$ (e.g., in genomic applications) it is reasonable to estimate $\sigma^{2}$ by the empirical variance of the residuals from the linear model with $\vec\beta = 0$.
}

\emph{
When $\vec y$ is not assumed to be normal, the maximum likelihood estimator for the regression coefficients has an asymptotic normal distribution $\sqrt{n}(\hat{\vec\beta} - \vec\beta) \rightarrow^D N(0, {{\bf I}^{-1}}(\vec\alpha, \vec\beta))$ 
and an asymptotic truncated-normal distribution post-selection. While ${{\bf I}^{-1}}(\vec\alpha, \vec\beta)$ depends on $\vec\beta$ and therefore cannot be assumed to be known in general, if $\vec X_1,\ldots,\vec X_m$ explain little of the variance in $\vec y$ it is reasonable to estimate the variance of $\hat{\vec\beta}$ under the assumption that $\vec\beta = 0$. More generally, \citet{tibshirani2015uniform} discuss conditions under which naive plug-in estimation $\widehat{{\bf I}^{-1}(\vec \alpha, \vec \beta)} = {\bf I}^{-1}(\hat{\vec\beta},\hat{\vec\alpha})$ leads to asymptotically valid inference. 
}
\end{example}

\section{Testing following selection}\label{sec-testing}

In the absence of selection, we can test for $H_j: \beta_j=0$ using the $p$-value of the test statistic $\hat \beta_j/SE_j$: $p_j = 2(1-\Phi(|\hat \beta_j/SE_j|))$ where $SE_j =  \sqrt{\vec e_j'{\bf \Sigma} \vec e_j}$ and $\vec e_{j}$ is the $m\times 1$ unit vector with a single entry of one in position $j\in \{1,\ldots,m\}$. However, conditionally on selection, $P_j$ will often have a distribution that is stochastically smaller than uniform, meaning that its realization $p_j$ will no longer be a valid $p$-value for testing $H_j$. 

To correct for selection, it appears necessary to evaluate the probability that $S\geq S_{1-t_1}$. However, this probability depends on the unknown $\vec \beta$, and hence it  cannot be evaluated when $H_j$ is true unless we assume that all other entries in $\vec \beta$ are zero. In the special case of $\vec \beta = \vec 0$ the distribution of $\hat \beta_j/SE_j$, conditional on $S\geq S_{1-t_1}$ is known. Of course, in practice we do not know whether any of the entries of $\vec \beta$ are non-zero.  

In \S~\ref{subsec-poly} we suggest a way around this problem, by computing a valid conditional $p$-values using the polyhedral lemma first introduced by \cite{Lee16}. In practice, we find that statistical tests based on the polyhedral lemma tend to have relatively low power if $\vec\beta$ is sparse, and thus, in \S~\ref{subsec-hybrid} we suggest an inference method that automatically adapts to the sparsity level of $\vec\beta$. In \S~\ref{subsec-mtp} we discuss applying multiple testing procedures to the valid conditional $p$-values.

\subsection{The conditional $p$-values based on the polyhedral lemma}\label{subsec-poly}
Performing inference on $\vec\beta$ is difficult because the post-selection distribution of $\hat{\vec\beta}$ is not location invariant and depends heavily on the unknown parameter value. Suppose that we are interested in performing inference on an arbitrary  linear function of the parameter vector $\vec{\eta}'\vec{\beta}$. \cite{Lee16} showed that by conditioning on additional information beyond the selection event, the post-selection distribution of a single coordinate $\hat\beta_j$ can be reduced to a univariate truncated normal distribution which only depends on single unknown parameter $\vec\eta'\vec\beta$. Furthermore, \cite{Fithian15} showed that such conditioning yields a (unique) family of admissible unbiased post-selection tests. 

Denote by $TN(\mu, \sigma^2, \mathcal A)$ the truncated normal distribution constrained to $\mathcal A \subseteq \Re$, i.e. the conditional distribution of a $N(\mu, \sigma^2)$ random variable conditional on it being in  $\mathcal A$.  Let $F^{\mathcal A}_{\mu, \sigma^2}$ be the CDF of $TN(\mu, \sigma^2, \mathcal A)$. The following theorem, which is a direct result of the polyhedral lemma of \cite{Lee16}, provides us with a conditional distribution of any linear contrast of $\hat{\vec \beta}$ that we can use for post-selection inference. 
\begin{theorem}\label{thm1}
Let $\vec \eta' \hat {\vec \beta}$ be a linear combination of $\hat {\vec \beta}$ , and $t_1\in (0,1]$ a fixed selection threshold. Let $\vec W = ({\bf I}_{m}-c{\vec \eta}'){\hat {\vec \beta}}$, where $c = ({\vec \eta}'{\bf \Sigma}{\vec \eta})^{-1}{\bf \Sigma} {\vec \eta}$ and ${\bf I}_{m}$ is the $m\times m$ identity matrix. 
Then 
\begin{equation}\label{eq-TN}
\vec \eta'\hat {\vec \beta} \mid S \geq S_{1-t_1} , \vec W \sim TN(\vec \eta'  \vec \beta, \vec \eta'{\bf \Sigma} \vec \eta, \mathcal A(\vec W)), 
\end{equation}
where  $\mathcal A(\vec W)$ is defined in Lemma \ref{lem1}. 
\end{theorem}
See Appendix \ref{app-thm1-proof} for the proof. 

Since the only unknown parameter in the truncated distribution of \eqref{eq-TN} is $\vec \eta'  \vec \beta$, it is straightforward to compute a $p$-value under the null hypothesis and construct confidence intervals via test inversion. 
\begin{cor}\label{cor1}
For the estimation of $\beta_j$, let 
 $\vec W= {\vec W_j} = ({\bf I}_{m}-c{\vec e_j}'){\hat {\vec \beta}}$, $c = ({\vec e_j}'{\bf \Sigma} {\vec e_j})^{-1}{\bf \Sigma} {\vec e_j}$, and  $\mathcal A(\vec W)$ as defined in Lemma \ref{lem1}. Then, 
$$
P\left( 
\beta_j \in \{b: \; \alpha / 2 \leq F^{\mathcal A}_{b, {\vec e_j}'{\bf \Sigma} {\vec e_j}}(\hat\beta_j) \leq 1 - \alpha /2\} 
\right) =1 - \alpha.
$$
For testing $H_j$, let
\begin{equation}\label{eq-condpv-poly}
P_j' = 1 - F^{\mathcal A}_{0, {\vec e_j}'{\bf \Sigma} {\vec e_j}}(\hat\beta_j).
\end{equation} Then, if $H_j$ is true, 
$$
P_j' | S> S_{1-t_1}, \vec W_j \sim U(0,1).
$$
\end{cor}
The following example serves to give some intuition as to how the polyhedral lemma works and the possible adverse effects of the extra conditioning on $\vec W$. 

\begin{example} Independence Model. \emph{
Let $\hat{\vec\beta}\sim N(\vec\beta,{\bf I}_m)$ and suppose that we are interested in testing $H_{1}: \beta_1 = 0$ after rejecting the global-null hypothesis that $\vec\beta = 0$. In this case, the relevant contrast is $\vec\eta = \vec e_1$ and the orthogonal projection is $\vec W = (0,\hat\beta_2,...,\hat\beta_m)$. It is clear that $\hat\beta_1$ is independent of $\vec W$ and therefore, conditionally on selection the only relevant information contained in $\vec W$ is that $\hat\beta_1 > S_{1-t_1} - \sum_{j=2}^{m}\hat\beta_j^{2}$ so $\mathcal{A} = \{b:\; b^{2} - S_{1-t_1} + \sum_{j=2}^{m}\hat\beta_j^{2} > 0\}$. Note that if we do not condition on $\vec W$ then the support of $\hat\beta_1| S>S_{1-t_1}$ is $\Re$ and this is why conditioning on $\vec W$ often results in a loss of efficiency \citep{Fithian15}.
}
\end{example}

\subsection{A hybrid conditional $p$-value}\label{subsec-hybrid}
Our empirical investigation in Section \S~\ref{sec-sim} suggests that $p$-values that are computed based on the polyhedral lemma tend to have good power when $\vec\beta$ is not sparse or has a large magnitude. However,  when only a single entry in $\vec \beta$ is nonzero, $p$-values based on the polyhedral lemma (which are valid for any configuration of the unknown $\vec \beta$) tend to be considerably less powerful than $p$-values computed based on the distribution of $\hat{\vec\beta}$ under the global-null distribution (where it is assumed that $\vec \beta= \vec 0$). Therefore, we would like to consider a test that adapts to the unknown sparsity of the signal, by combining  the two approaches for computing $p$-values into a single test of $H_j$, allowing for powerful identification of the non-null coefficients. The combined test will be useful in applications where both groups with sparse signals and with non-sparse signals are likely.

Sampling from the truncated multivariate normal distribution is a well studied problem, see for example \cite{Pakman14}. Specifically, under the global null, i.e., $\vec \beta = \vec 0$, one can use samples from the truncated distribution to asses the likelihood of the observed regression coefficients, defining
\begin{equation}\label{eq-condpv-GN}
p'_{j,GN} =Pr_{\vec \beta  = \vec 0}\left(P_j\leq p_j \mid S>S_{1-t_1}\right) = \frac{1}{t_1}Pr_{\vec \beta  = \vec 0}\left(P_j\leq p_j , S>S_{1-t_1}\right),
\end{equation}
$j=1,\ldots, m.$  
Under the global-null distribution, both $P'_j$ and the $p$-value computed under the global-null distribution $P'_{j,GN}$ have a uniform distribution. However, $p'_j$ may be larger than $p'_{j, GN}$ because it requires extra conditioning on $\vec W$. Thus, if the only non-zero predictor in the model is the $j$th predictor, the test based on $P'_{j,GN}$ can be expected to be more powerful than the test based on $P'_j$. 

On the other hand, when more than one of the coordinates of $\vec\beta$ are non-zero, $p'_{j,GN}$ will often be substantially larger than the original $p$-value $p_j$, e.g., if $\hat \beta_j^2/SE_j^2\geq S_{1-t_1}$, then $p'_{j,GN} = p_j/t_1$. This, while $P'_j$ may not suffer any additional loss of power due to the extra conditioning  e.g., if the aggregate test passes the selection threshold $t_1$  regardless of the value of $p_j$, then $p'_j = p_j$ and it will clearly be smaller than $p'_{j,GN}$.

Since the preference for using $p'_{j,GN}$ instead of $p'_j$ depends on the (unknown) $\vec \beta$, we suggest the following  test that combines the two valid post-selection $p$-values, 
\begin{equation}\label{eq-condpv-hybrid}
p'_{j, hybrid} = 2\min (p'_{j}, p'_{j, GN}).
\end{equation}

Clearly, $p'_{j, hybrid}$ would be a valid $p$-value, i.e., with a null distribution that is either uniform or stochastically larger than uniform,  if both $p'_j$ and $p'_{j,GN}$ are valid $p$-values. In the previous section we indeed showed that $p'_j$ is a valid $p$-value. But by the definition in equation \eqref{eq-condpv-GN} it is only clear that $p'_{j,GN}$ is valid when $\vec \beta=0$. Intuitively, for $\vec \beta \neq 0$, we may assume that $p'_{j,GN}$ is conservative (i.e., has a null distribution that is stochastically larger than uniform).  We shall now provide a rigorous justification.

We start with the special case that the quadratic aggregate test for selection is Wald's test. Following selection by Wald's test, $P'_{j, hybrid}$ is a valid $p$-value for testing $H_j: \beta_j = 0$. This follows by showing that  the marginal null distribution of $P'_{j,GN}$  is at least stochastically as large as the uniform, so the test based on the global null distribution where  $\vec \beta = \vec 0$ is conservative.  
\begin{theorem}\label{thm-WALD-GN}
If $S = \hat {\vec \beta}' {\bf \Sigma}^{(-1)}\hat {\vec \beta}$, and $\hat{\vec \beta}$ has a normal distribution with mean $\vec \beta$ and variance ${\bf \Sigma}$, then $$Pr_{(\beta_1,\ldots,\beta_{j-1},0,\beta_{j+1}, \ldots, \beta_m)}\left( P'_{j,GN}\leq x \right)\leq x \quad \forall x\in [0,1].$$
\end{theorem}
See Appendix \ref{app-thm-WALD-GN-proof} for the proof.

More generally, when selection is based on $S =\hat {\vec \beta}'{\bf K}\hat {\vec \beta}> S_{1 - t_1}$, where ${\bf K}$ is any positive definite symmetric matrix, we can still justify the use of $p'_{j, hybrid}$ for testing $H_j: \beta_j = 0$ for a large enough sample size. This follows since the conditional $p$-values under the global null are necessarily larger than the original $p$-values, as formally stated in the following lemma. 
\begin{lemma}\label{lemma-gaussiancorrelationineq}
If $\bf K$ is a positive definite matrix, $S =\hat {\vec \beta}'{\bf K}\hat {\vec \beta}$, and $\hat {\vec \beta}\sim N(\vec \beta, {\bf \Sigma})$, then 
\begin{equation}\label{eq-lemma-gaussiancorrelationineq}
Pr_{\vec \beta = \vec 0}(\hat\beta_j^{2} > b | S > s) 
\geq  Pr_{ \beta_j =  0}(\hat\beta_j^{2} > b)
\end{equation}
for arbitrary fixed $b, s > 0$.

\end{lemma}
See Appendix \ref{app-lemma-gaussiancorrelationineq} for the proof. 
Setting $b$ to be the realized test statistic and $s = S_{1-t_1}$, $p'_{j,GN}=Pr_{\vec \beta  = \vec 0}\left(P_j\leq p_j \mid S>S_{1-t_1}\right)$ is the lefthand side of \eqref{eq-lemma-gaussiancorrelationineq} and $p_j = Pr(\chi^2_1\geq \hat \beta_j^2/SE_j^2)$ is the righthand side. It thus follows that 
$$p'_{j,GN}\geq p_j.$$ 
Since $\lim_{n\rightarrow\infty} P_{(\beta_1,\ldots,\beta_{j-1},0,\beta_{j+1}, \ldots, \beta_m)}(S > S_{1 - t_1})=1$ regardless of the true value of $\beta_j$ if $\beta_k\neq 0$ for at least one $k\neq j$, the probability of getting a smaller value than $p_j$ given selection, if $H_j$ is true, coincides with $p_j$ asymptotically. So $p_j$ is an asymptotically valid $p$-value if $\beta_k\neq 0$ for at least one $k\neq j$.  Since $p'_{j,GN}\geq p_j$,  it follows that $p'_{j,GN}$ and $p'_{j, hybrid}$ are asymptotically valid $p$-values for any $\vec \beta$.  

%SHOW THAT IF THE GLOBAL NULL IS TRUE, THE HYBRID PV IS CLOSE TO UNIFORM. 
\subsection{Controlling the conditional error  rate}\label{subsec-mtp}
In order to identify the non-null entries in $\vec \beta$, we can apply a valid multiple testing procedure on the conditional $p$-values computed as in \S~\ref{subsec-poly} or \S~\ref{subsec-hybrid}. We can then achieve conditional error control. 

The Bonferroni-Holm procedure will control the conditional FWER,  since the conditional $p$-values are valid $p$-values and the procedure is valid under any dependency structure among the test statistics. 

For conditional FDR control, we recommend using the Benjamini-Hochberg (BH) procedure. Although the BH procedure does not have proven FDR control for general dependence among the $p$-values, it usually controls the FDR for dependencies encountered in practice.  We believe that the robustness property of the BH procedure carries over to our setting, and that the conditional FDR will be controlled in practice. The robustness guarantee follows from empirical and theoretical results \citep{Reiner07}, which suggest that  the FDR of the BH procedure does not exceed its nominal level for test statistics with a joint normal distribution, and our simulations in \S~\ref{sec-sim}, which suggest that this holds also following selection.  

A conservative procedure that will control the conditional FDR is the Benjamini-Yekutieli procedure for general dependence, introduced in \cite{Benjamini01}. The theoretical guarantee follows since the conditional $p$-values are valid $p$-values and the procedure is valid under any dependency structure among the test statistics.

If the individual test statistics are independent, as occurs when the design matrix $\bf X$ is orthogonal in the linear model, and the aggregate test statistic is monotone increasing in the absolute value of each test statistic (keeping all others fixed), then we have a theoretical guarantee that the BH procedure on $p'_1,\ldots,p'_m$ controls the conditional FDR, even though these conditional $p$-values are dependent. This is a direct result of Theorem 3.1 in \cite{Heller16}, and it is formally stated in the following theorem. 
\begin{theorem}\label{thm-WALD-BH}
If $S = \hat {\vec \beta}' {\bf \Sigma}^{(-1)}\hat {\vec \beta}$, $\hat{\vec \beta}\sim N(\vec \beta,{\bf \Sigma})$, and $\bf \Sigma$ is a diagonal matrix, then the BH procedure at level $\alpha$ on $p_1',\ldots,p_m'$ controls the conditional FDR at level $m_0/m\times \alpha$, where $m_0$ is the number null coefficients in $\vec \beta$. 
\end{theorem}

\section{Estimation following selection}\label{sec-estimation}
So far we focused on  valid testing after selection by an aggregate test. But it is often also desirable to asses the absolute magnitude of parameters of interest. Just as model selection causes an inflation of test statistics, it also has an adverse effect on the accuracy of point estimates. In fact, inflation of estimated effect sizes is the main cause for the increased type-I error rates that are encountered in naive inference following selection. In Section \S~\ref{sub-computation} we discuss the computation of post-selection of maximum likelihood estimators which are defined as the maximizers of the likelihood of the data conditional on selection and serve to correct for some of the selection bias. %In Section \S~\ref{sub-mletheroy} we a give a consistency statement for the conditional maximum likelihood estimates. 

Beyond point estimates, valid post-selection confidence intervals can be constructed by inverting the post-selection tests described in Section \S~\ref{sec-testing}. These however, may  be either underpowered in the case of confidence intervals based on the polyhedral lemma or too conservative in the case of the hybrid confidence intervals. Thus, in Section \S~\ref{sub-CI} we propose novel regime switching confidence intervals that maintain the validity and power of the hybrid method intervals while ensuring the desired level of confidence asymptotically.

\subsection{Conditional maximum likelihood estimation}\label{sub-computation}
Let $\ell(\vec \beta)$ be  the log-likelihood for $\vec \beta$, and $\ell(\vec \beta| S > S_{1-t_1})$ the corresponding conditional log-likelihood. 
Define the conditional MLE as the maximizer of the conditional likelihood:
\begin{equation}\label{conditionalMLE}
\tilde{\vec\beta} = \arg\max_{\vec \beta} \ell(\vec \beta) -  \log Pr_{\vec \beta}( S > S_{1-t_1}).
\end{equation}
For notational convenience, we suppress the dependence of $\tilde{\vec\beta}$ on the selection threshold $t_1$. While difficult to compute in many practical cases, computing the conditional MLE following selection by aggregate testing is a relatively simple task. For the special case where $\bf K = \bf\Sigma^{(-1)}$, we are able to show that the maximum likelihood estimator is given by the solution to a simple line search problem.
\begin{theorem}\label{thm:linesearch}
Under the conditions of Theorem \ref{thm-WALD-GN}, the conditional maximum likelihood estimator is given by:
\begin{align*}
\tilde{\vec\beta} &= \arg\max_{\vec\beta} \ell(\vec\beta) - \log Pr_{\vec\beta}(S > S_{1-t_1})  \\
&= \arg\max_{\lambda\in[0, 1]} \ell(\lambda\hat{\vec\beta}) - \log Pr_{\lambda\hat{\vec\beta}}(S > S_{1-t_1})
\end{align*}
where $\hat{\vec\beta}$ is the observed value. 
\end{theorem}
See appendix \ref{app-linesearch} for the proof. 

The Theorem shows that the maximum likelihood estimation is reduced to maximizing the likelihood only with respect to a scalar factor. This follows when $\bf K=\bf\Sigma^{(-1)}$  because the distribution of the test-statistic is governed by one unknown parameter. In the general case, the distribution of $S$ is a sum of chi-square random variables which depends on $\text{rank}(\bf K)$ parameters, making the optimization problem slightly more involved. So, for ${\bf K} \neq {\bf \Sigma^{-1}}$ we use the stochastic optimization approach of \citet{Meir17} to maximize the likelihood. Let 
$$
\vec z(\vec\beta) \sim f_{\vec\beta}(\hat{\vec\beta}|S > S_{1-t_1})
$$ 
be a sample from the post selection distribution of $\hat{\vec\beta}$ for a mean parameter value $\vec\beta$. Then, taking gradient steps of the form
\begin{equation}\label{eq-grad-steps}
\tilde{\vec\beta}^{t+1} = \tilde{\vec\beta}^t + \gamma_t{\bf\Sigma}^{(-1)}\left(\hat{\vec\beta} - \vec z(\tilde{\vec\beta}^{t})\right)
\end{equation}
will lead to convergence to the conditional MLE as long as
$$
\sum_{t=1}^{\infty} \gamma_t = \infty \qquad \text{and}
\qquad \sum_{t=1}^{\infty}\gamma_t^{2} < \infty.
$$
\begin{theorem}
Suppose that $\hat{\vec\beta}\sim N(\vec\beta, {\bf\Sigma})$ and that inference is conducted only if  $S >  S_{1-t_1}$ with $S = \hat{\vec\beta}'{\bf K} \hat{\vec\beta}$. Then, the algorithm defined by \eqref{eq-grad-steps} converges to the conditional MLE for the post-aggregate testing problem which satisfies
$$
\lim_{t\rightarrow\infty} \hat{\vec\beta} - E_{\tilde{\vec\beta}^t}(\hat{\vec\beta} | S > S_{1-t_1}) = 0. 
$$
\end{theorem}
\begin{proof} The result follows from the fact that the variance of the post-selection distribution of $\hat{\vec\beta}$ can be uniformly bounded from above by ${\bf\Sigma} / t_1$, see \citet{Meir17} for details. \end{proof}

We discuss the topic of conditional maximum likelihood estimation in Generalized Linear Models and the related problem of estimation after aggregate testing with a linear test in Appendices \ref{appendix-mle-computation} and \ref{sec-linear-aggregate-test}.

The conditional MLE is consistent assuming the following. 
Suppose that we observe a sequence of regression coefficient estimates $\hat{\vec\beta}_1,\dots,\hat{\vec\beta}_n,\dots$ such that
\begin{equation}\label{eq-beta-condition}
%\sqrt{n}(\hat{\vec\beta}_n - \vec\beta) \sim N(0, \bf\Sigma_n), \qquad \lim_{n\rightarrow\infty} \bf\Sigma_n =^P \bf\Sigma.
\hat{\vec\beta}_n  \sim N(\vec\beta, {\bf\Sigma}_n), \qquad n{\bf\Sigma}_n \textrm{ converges in probability}.
\end{equation}
Furthermore, suppose that we perform inference on the individual coordinates of $\hat{\vec\beta}_n$ if and only if 
\begin{equation}\label{eq-test-condition}
S_n > S_{1 - t_1}, \qquad S_n = \hat{\vec\beta}'_n {\bf K_n} \hat{\vec\beta}_n.
\end{equation}
The good behaviour of the conditional MLE hinges on the probability of passing the selection by the aggregate test. The lower bound on the this probability is given trivially by $t_1$ and therefore the conditional MLE is consistent. 
\begin{theorem}\label{thm-const-asymp}
Assume that \eqref{eq-beta-condition} and \eqref{eq-test-condition} hold. Then, the conditional MLE is consistent for $\vec\beta$, satisfying:
$$
\lim_{n\rightarrow\infty} Pr\left(\|\tilde{\vec\beta}_n - \vec\beta\|_{\infty}  > \varepsilon | S_n > S_{1-t_1}\right) = 0, \qquad \forall\varepsilon > 0. 
$$
\end{theorem}
\proof The result follows from the theory developed in the work of \citet{Meir17} for selective inference in exponential families and the fact that
$$
\inf_{\vec\beta} Pr_{\vec\beta}(S_n > S_{1-t_1}) = t_1, \qquad \forall n.
$$
\qed

For a discussion of post-selection efficiency, see the work of \citet{Routenberg15}.

\subsection{Confidence intervals following selection by an aggregate test}\label{sub-CI}
From Theorem \ref{thm1} it is clear that the truncated normal distribution can be used to construct confidence intervals post-selection in a straightforward manner. However, the extra conditioning (on $\vec W$) may lead to wide confidence intervals relative to confidence intervals based on the sampled distributions, as pointed out by \cite{Tian16}. As an alternative, it is possible to  invert a global type test (specifically, the test with null hypothesis $\vec \beta = \vec e_j b$ for coefficient $\beta_j$) and construct a hybrid type confidence interval in order to obtain a confidence interval with more power to determine the sign of the regression coefficients \citep{Weinstein13}.

For constructing a confidence interval at a $1 - \alpha$ level, let $L'_j(\alpha)$ and $U'_j(\alpha)$ be the lower and upper bounds of the polyhedral confidence interval for the $j$th variable, so:
$$
F^{\mathcal{A}}_{L',{\vec e_j}'{\bf \Sigma} {\vec e_j}}(\hat\beta_j) = 1 - \frac{\alpha}{2}, \qquad 
F^{\mathcal{A}}_{U',{\vec e_j}'{\bf \Sigma} {\vec e_j}}(\hat\beta_j) = \frac{\alpha}{2},
$$
where $F^{\mathcal{A}}_{b,\sigma^2}$ is as defined in \S~\ref{subsec-poly}. Similarly, let $L_{GN,j}'(\alpha)$ and $U_{GN,j}'(\alpha)$ be the lower and upper limit of the global-null confidence interval for the $j$th variable:
$$
\left\{b: \; \alpha/2 \leq F_{\vec\beta = \vec e_j b}(\hat\beta_j| S > S_{1-t_1}) \leq 1 - \alpha/2\right\},
$$
where $\vec e_j$ is the unit vector and $F_{\vec\beta = \vec e_j b}(\hat\beta_j| S > S_{1-t_1})$ is the CDF of ${\vec e}_j'\hat {\vec \beta}$ given selection, for the parameter vector $\vec\beta = \vec e_j b$. We use the Robbins-Monroe process to find $L_{GN,j}'$ and $U_{GN,j}'$ \citep{Garthwaite92}. As in testing, the polyhedral confidence interval tends to be shorter and more efficient if there are several variables in the model that are highly correlated with the response variable and the global-null confidence intervals tend to be more powerful when the model is sparse or if the global-null hypothesis holds (approximately). As we have done in Section \S~\ref{subsec-hybrid}, we propose a hybrid method for constructing a confidence interval, as defined by the lower and upper bounds:
$$
L_{hybrid,j}(\alpha) = \max\{L'_j(\alpha/2), L'_{GN,j}(\alpha/2)\}, \qquad
U_{hybrid,j}(\alpha) = \min\{U'_j(\alpha/2), U'_{GN_j}(\alpha / 2)\}.
$$

The hybrid confidence intervals, while possessing a good degree of power to determine the sign regardless of the true underlying model, tend to be inefficient when there is strong signal in the data. To see why, consider the case of a regression model where $\beta_1,\beta_2 > 0$. Then, for a sufficiently large sample size the polyhedral confidence interval will apply no correction and hybrid confidence interval will be conservative, with an asymptotic level of $1-\alpha/2$: $\lim_{n\rightarrow \infty} P\{(L',U')_{j,hybrid}(\alpha) = (L',U')_{j}(\alpha / 2)\} = 1$. As a remedy, we propose a regime switching scheme for constructing confidence intervals in which we first determine whether $\|\vec\beta\|\approx 0$ or $\|\vec\beta\| \gg 0$ and then construct confidence intervals accordingly.

\begin{procedure}\label{proc-postselectionCI}
The post-selection level $1-\alpha$ confidence interval for $\beta_j$, with  switching regime at level $t_2<\alpha\times t_1$ (with default value $t_2 = \alpha^{2}\times t_1$):
\begin{enumerate}
\item Compute $S_{1-t_2}> S_{1-t_1}$. 
\item If $S<S_{1-t_2}$, i.e., the aggregate test does not pass the more stringent threshold $t_2$, then compute  the hybrid  conditional confidence interval at level $1-\alpha^* =1-(\alpha-t_2/t_1)$.
\item If $S\geq S_{1-t_2}$, compute  the  unconditional confidence interval, at level $1-\alpha^*=1- \alpha$. 
  \end{enumerate}
\end{procedure}

\begin{theorem}\label{thm-ci}
Post-selection confidence intervals constructed with Procedure \ref{proc-postselectionCI} have a confidence level at least $1-\alpha$ if $\vec \beta = \vec 0$, and an asymptotic level $1-\alpha$ if $\vec \beta \neq 0$. 
\end{theorem}
See Appendix \ref{app-pf-thm-ci} for the proof. 

\begin{remark}\emph{
Ours is not the first regime switching procedures proposed for inference in the presence of data-driven variable selection, see for example the works of \citet{Chatterjee2011} and \citet{McKeague2015}. In both these cases, one has to determine whether some (or all) of the parameters are zero and construct a test in an appropriate manner. The usual prescription for selecting tuning parameters in such procedures is to scale the tuning parameter of the test (in our case, $t_2$) in such a way so the correct regime is selected with probability approaching one as the sample size grows. In our case, this would amount to setting $t_{2,n}$ in such a way so that $t_{2,n}\rightarrow 0$ and $S_{1 - t_{2,n}} = o(n)$. However, in practice it is necessary to select a single a value for $t_2$ and so we chose to fix $t_2$ to a small value as to maintain a good degree of power when there is only limited amount of signal in the data and to modify our procedure in such a way as to ensure some finite sample coverage guarantees. 
}\end{remark}

\begin{figure}[t]
\label{fig-est-ci}
\begin{center}
\includegraphics[width= 14 cm]{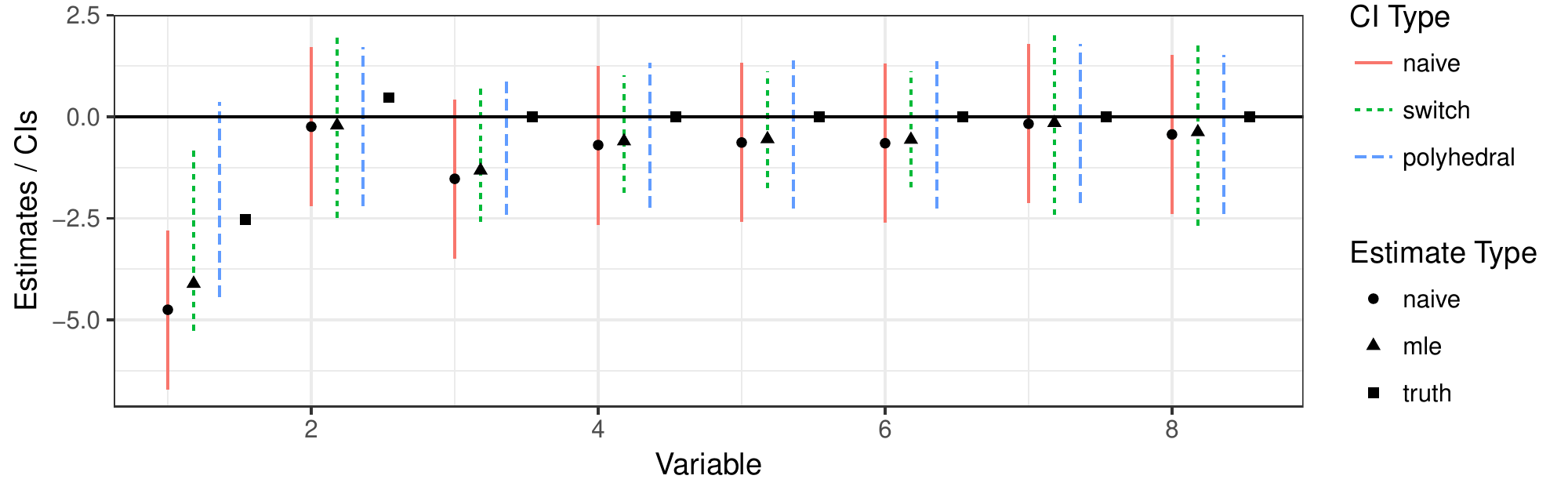}
\caption{Point estimates and confidence intervals for the artificial data example described in Example \ref{ci-example}. The naive point estimate is marked as a circle, the conditional MLE as a triangle and the true value of the parameter is marked as a square. The confidence intervals are the naive (solid red line), hybrid (dotted green line) and polyhedral (dashed blue line).}
\end{center}
\end{figure}

\begin{example}\label{ci-example}\emph{
Figure \ref{fig-est-ci} shows point estimates and confidence intervals for normal means vector which was selected via a quadratic aggregate test. The figure was generated by sampling $\hat{\vec \beta} \sim N_8(\vec\beta,{\bf \Sigma})$ with $\Sigma_{i, j} = 0.3 I_{i\neq j} + 1 I_{i = j}$, $\beta_1 = -2.5$, $\beta_2 = 0.5$ and $\beta_3 = \dots =\beta_8 = 0$. The aggregate test applied was a Wald test at an $\alpha = 0.001$ level. The naive and conditional estimates are plotted along with naive, polyhedral and hybrid $95\%$ confidence intervals. The conditional MLE applies the same multiplicative shrinkage of $0.86$ to all of the coordinates of $\hat{\vec \beta}$ and so the shrinkage is more visible for the larger observed values. Because the selection is driven by $\hat\beta_1$ corresponding to the large negative coordinate $\beta_1$, the polyhedral confidence intervals for the other coordinates of $\beta$ are similar in size to the naive confidence intervals. The naive confidence intervals overestimates the magnitude of $\beta_1$, the polyhedral confidence intervals cover the true parameter value but fails to determine its sign and the regime switching confidence intervals both cover the true parameter value and succeed in determining the sign. 
}\end{example}

\section{Simulations}\label{sec-sim}
In this section we conduct a simulation study where we assess the methods proposed in this work and verify our theoretical findings. In Section \S~\ref{sub-sim-testing} we asses the post-selection tests proposed in Section \ref{sec-testing} with respect to their ability to control the FDR. In Section \S~\ref{sub-power} we compare the different testing method with respect to their power to detect true signal in the data. In Section \S~\ref{sub-estimation} we compare the conditional MLE and the unadjusted MLE with respect to their estimation error. Finally, in Section \S~\ref{sub-cover} we asses the coverage rates of the polyhedral and regime switching confidence intervals. 

In all of our simulations we generate data in a similar manner. We first generate a design matrix in a manner meant to approximate a rare-variant design. We sample marginal expression proportions for our variants from $g_1,...,g_m \sim Gamma(1, 300)$ constrained to $[2\times 10^{-4},0.1]$ and for each subject we sample two multivariate normal vectors $r_{i,.} \sim N(0, U)$ with $U_{ij} = 0.8^{|i-j|}$. We then set $X_{i,} = \sum_{k=1}^{2} I\{\Phi(r_{i,j}) \leq g_j\}$ to obtain a design matrix with dependent columns and a marginal distribution $X_{ij}\sim Bin(2, g_{j})$. We generate a sparse regression coefficients vector with $m - s$ zero coordinates and $s$ coordinates which are sampled form the $ \text{Laplace}(1)$ distribution. We normalize the values of the regression coefficients such that the signal to noise ratio
\begin{equation}\label{eq-snr-def}
\text{snr} = \sqrt{\vec\beta'(X'X)\vec\beta}
\end{equation}  
equals some pre-specified value. Finally, we generate a response variable $y = X\beta + \varepsilon$ with $\varepsilon\sim N(0,I)$. In all of our simulations we use a Wald aggregate test with a significance level of $t_1 = 0.001$.
 
\begin{figure}[t]
\begin{center}
\includegraphics[width= 14 cm]{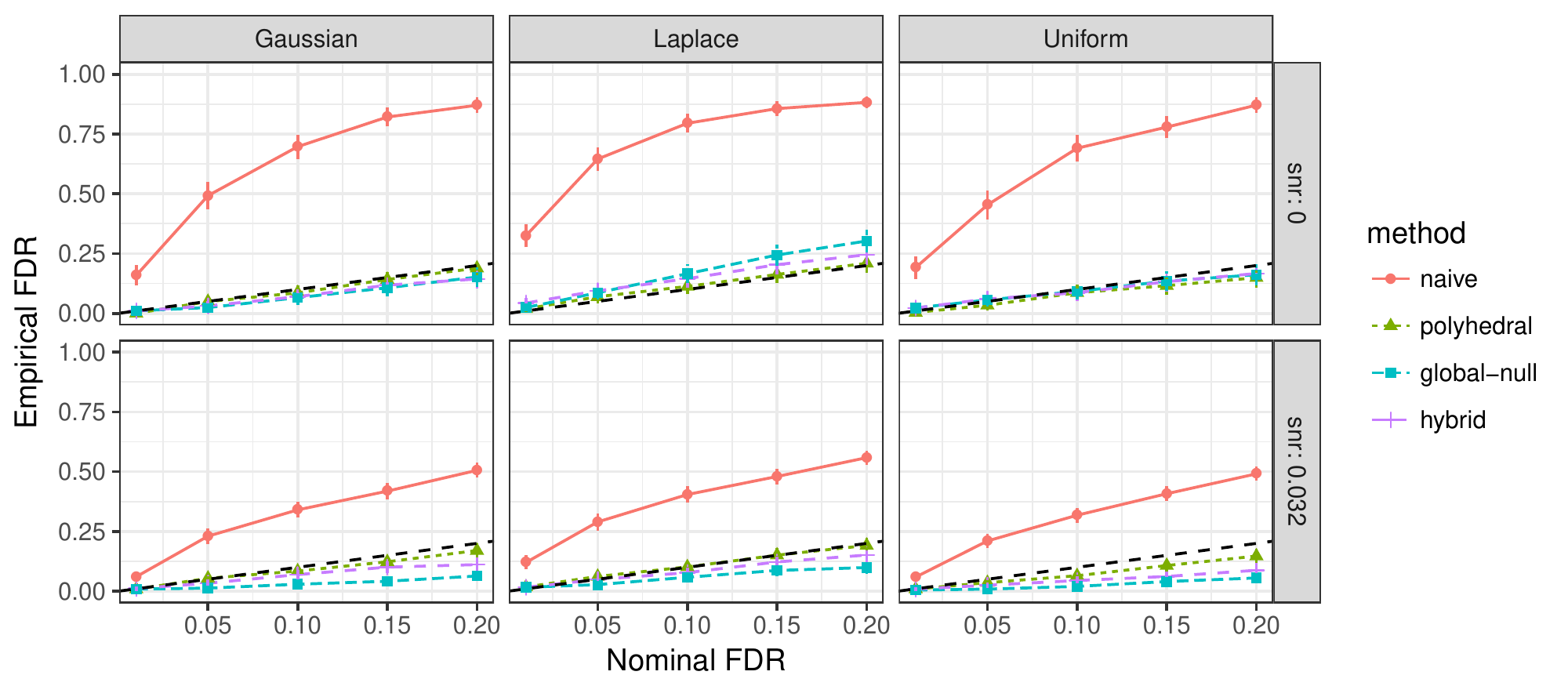}
\caption{False Discovery Rates after aggregate testing. We plot the nominal FDR vs. the empirical FDR for the unadjusted naive p-values (red solid line), the polyhedral p-values (dotted green line), p-values based on the exact post-selection null distribution (dashed-blue) and the hybrid method (dashed purple line). The diagonal line is in dashed black. The figure is faceted according to the distribution of the noise and the signal to noise ratio in the data as defined in equation \eqref{eq-snr-def}. Details about the data generation are in \S~\ref{sub-sim-testing}.}
\label{fig-fdr-sim}
\end{center}
\end{figure}

\subsection{Assessment of false discovery rate control}\label{sub-sim-testing}
We assess how well the proposed testing procedures control the FDR under the assumed model as well as under model misspecification. We generate datasets with $m = 50$, $n = 10^4$, $s = 3$, $\text{snr}\in\{0, 0.032\}$ and three types of distributions for the model residuals, all of which have a variance of $1$:
$$
\varepsilon^{(1)}_i \sim N(0, 1), \qquad 
\varepsilon^{(2)}_i \sim Laplace(\sqrt{2}), \qquad
\varepsilon^{(3)}_i \sim Unif\left(-\sqrt{12}/2, \sqrt{12}/2\right).
$$
We compare four testing procedures. BH on naive $p$-values which are not adjusted for selection, BH on the polyhedral $p$-values as computed in equation \eqref{eq-condpv-poly}, BH on the $p$-values based on the global null distribution as computed in equation \eqref{eq-condpv-GN}, and BH on the hybrid $p$-values as computed in equation \eqref{eq-condpv-hybrid}.

We plot the target FDR versus the empirical FDR in Figure \ref{fig-fdr-sim}. When there is no signal in the data and the noise is not heavy tailed, all selection adjusted methods obtain close to nominal FDR levels (top left and right panels). When the noise is heavy tailed (Laplace), the methods based on the null Gaussian distribution have higher than nominal FDR rates, while the p-values computed with the polyhedral method exhibit a more robust behavior (top center panel). When there is some signal in the data, all selection adjusted p-values control the FDR at nominal or conservative rate (bottom row). The naive p-values do not control the FDR in any of the simulation settings. Thus, we conclude that the polyhedral $p$-values  may be preferable to the hybrid and global null $p$-values if the distribution of the data is heavy-tailed. However, as we show in the next section, the hybrid method tends to have more power compared to the polyhedral method and is preferable when the residual distribution is well behaved.

\subsection{Assessment of power to detect true signal}\label{sub-power}
We compare the power to detect signal of the proposed testing procedures. We generate datasets with $m = 50$, $n = 10^4$, $s \in \{1, 2, 4, 8\}$, $\text{snr} \in\{0.032, 0.064, 0.128, 0.256\}$ and $\varepsilon_i \sim N(0, 1)$. We compare the same testing procedures as in \S~\ref{sub-sim-testing}. We measure the power to identify true signals at a nominal FDR level of $0.1$. 

We plot the results of the simulation in Figure \ref{fig-power-sim}.  In all of the simulations the naive unadjusted p-values have the most power, at the cost of an inflated FDR. When the number of non-zero regression coefficients is small, the global null and hybrid methods tend to have the most power, while all methods have a similar power when the signal is spread over a large number of regression coefficients.  
The method based on the global null distribution is the most powerful when the signal is sparse and low. The polyhedral method has more power when the signal is not too sparse or low. The hybrid method seems to adapt to the sparsity and signal strength well, exhibiting comparatively good power in all settings. 

 \begin{figure}[t]
\begin{center}
\includegraphics[width= 14 cm]{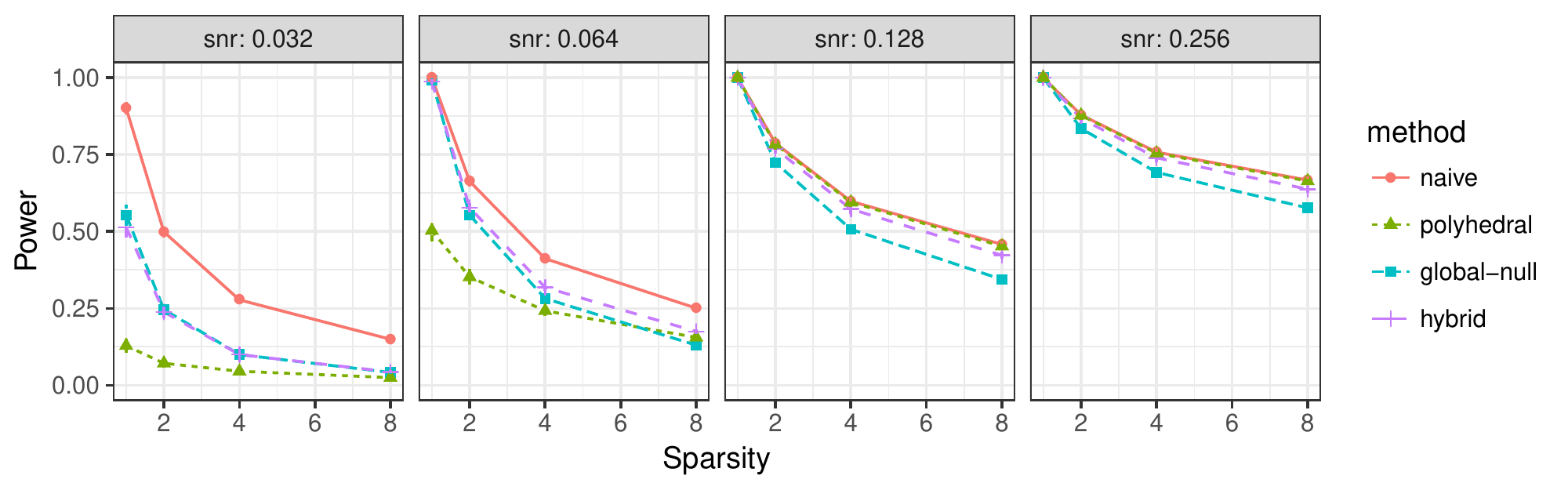}
\caption{Power to detect true signals after aggregate testing. We plot the power of the different inference method as a function as the number of non-zero regression coefficients for the unadjusted naive p-values (red solid line), the polyhedral p-values (dotted green line), p-values based on the exact post-selection null distribution (dashed-blue) and the hybrid method (dashed purple line).  The figure is faceted according to the strength of the signal as defined in equation \eqref{eq-snr-def}. Details about the data generation are in \S~\ref{sub-power}.}\label{fig-power-sim}
\end{center}
\end{figure}

\subsection{Assessment of estimation error}\label{sub-estimation}
We compare the conditional MLE to the naive, unadjusted point estimate, $\hat{\vec\beta}$ itself. We set $n \in \{5000, 10000, 15000, 20000\}$, $m \in \{5, 10, 20\}$, $s = 2$, $\text{snr} \in \{0, 0.025\}$ and sample model residuals from the normal distribution with a standard deviation of $1$

We plot the results of the simulation in Figure \ref{fig-rmse-sim}. When the dimension of $\vec\beta$ is small, the conditional MLE estimates the vector of regression coefficients better than the unadjusted MLE. The gap between the conditional and naive estimator is roughly constant across the different sample sizes when $\vec\beta = \vec 0$ because the probability of selection remains constant for all sample sizes. However, when there is some signal in the data the probability of passing the aggregate increases in the sample size and the gap between the estimators shrinks. The difference between the conditional MLE and the naive MLE decreases in the size of $\vec\beta$, to the extent that for $m = 20$ the two estimators are indistinguishable from one another. 

To see why this occurs, consider the following example. Let $y \sim N_m(0, I)$, suppose that we perform selection using a Wald test at a fixed level $t_1$ and consider the conditional likelihood function:
$$
\mathcal{L}(y) \propto -\frac{1}{2}\sum_{i=1}^{m}(y_i - \mu_i)^{2}
- \log P_\mu(S > S_{1-t_1}).
$$
As we let the dimension $m$ grow, the decrease in the value of the (unconditional) gaussian log-likelihood due to a possible shrinkage of $\mu$ grows linearly in $m$. At the same time, the additional penalty term $-\log P_{\mu} (S > S_{1 - t_1})$ remains bounded below by $-\log t_1$ regardless of the dimension of the problem.

\begin{figure}[t]
\begin{center}
\includegraphics[width= 14 cm]{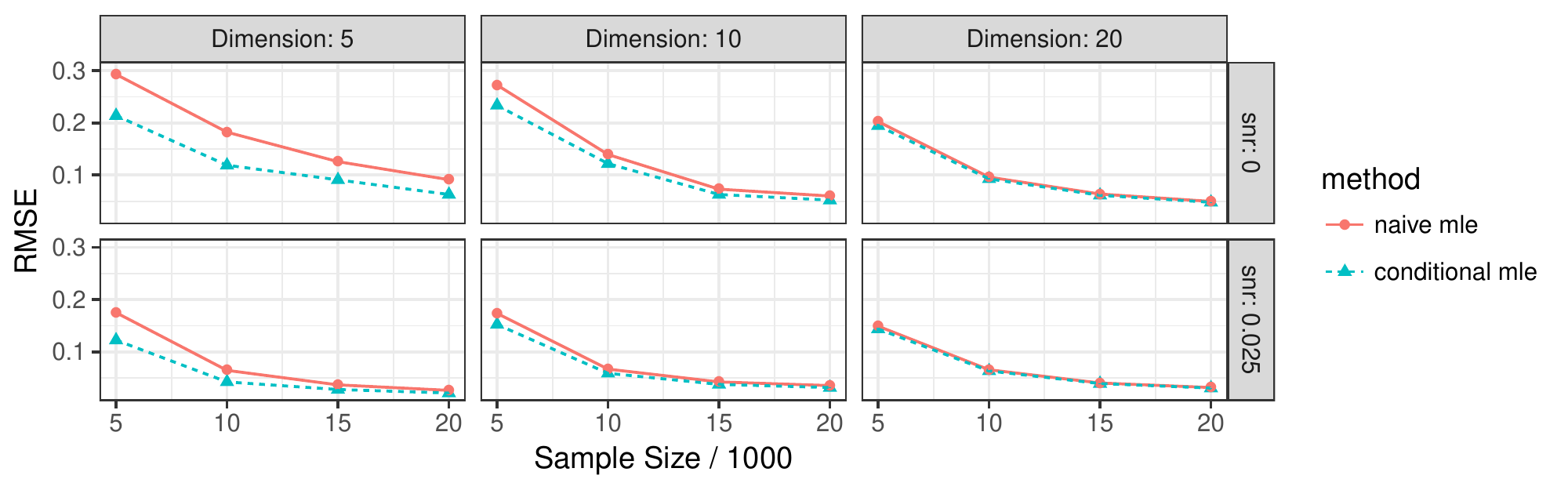}
\caption{Root mean squared error for estimation after aggregate testing. We plot the RMSE for estimating the vector of regression coefficients $\vec\beta$ with the naive unadjusted estimator $\hat{\vec\beta}$ (solid red line) and the conditional MLE $\tilde{\vec\beta}$ (dashed blue line). The figure is faceted according to the signal-to-noise ratio as defined in \eqref{eq-snr-def} and the size of $\vec\beta$, $m$.}\label{fig-rmse-sim}
\end{center}
\end{figure}

\subsection{Assessment of confidence interval coverage rates}\label{sub-cover}
In the last set of simulations, we evaluate the regime switching and polyhedral confidence intervals with respect to their coverage rates and power to determine the sign of the non-zero coefficients. We set the parameters of the simulation to $m = 20$, $n = 10^4$, $s \in \{1, 2, 4, 8\}$, $\text{snr}\in\{0.001, 0.002, 004, 008, 0.016, 0.032, 0.064, 0.128, 0.256\}$ and sample the residuals from a normal distribution with a standard deviation of $1$. 

We plot the results of the simulation in Figure \ref{fig-cover-sim}. The naive confidence intervals have a coverage rate far below nominal for signal-to-noise ratios less than $1$. As could be expected, the polyhedral method achieves the correct coverage rates up to Monte-Carlo error in all simulation settings. When there is no signal in the data the regime switching confidence intervals have close to nominal coverage. When the signal to noise ratio is moderate, the regime-switching confidence intervals are conservative because the polyhedral confidence intervals are superior to the ones based on the global-null assumption with high probability while the probability of $S$ exceeding $S_{1-t_2}$ is still not overwhelmingly large. When the signal to noise ratio is high the regime switching confidence intervals are mostly identical to the naive ones,  because the selection occurs with probability of close to $1$, and so they have the correct coverage rate. Despite being more conservative than the polyhedral confidence intervals, the regime switching confidence intervals can have better power to determine the sign. Specifically, the regime switching confidence intervals tend to have more power when the true model is sparse and the signal to noise ratio is low or moderate. 

\begin{figure}[t]
\begin{center}
\includegraphics[width= 14 cm]{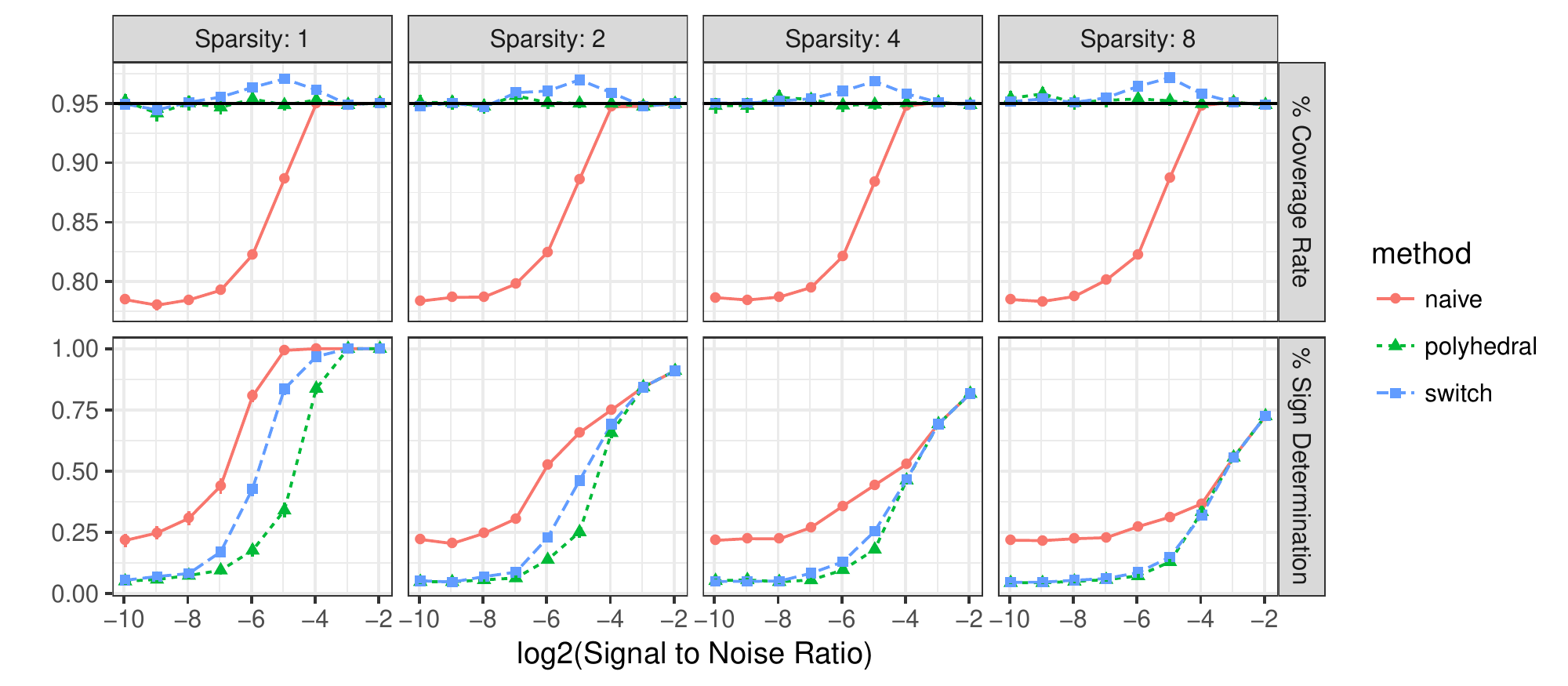}
\caption{Coverage rates and power to determine the sign of confidence intervals constructed after aggregate testing. We plot results rates for the naive unadjusted confidence intervals (solid red line), polyhedral confidence intervals (dotted green line) and regime switching intervals with $t_2 = t_1\alpha^2$ (dashed blue line).}\label{fig-cover-sim}
\end{center}
\end{figure}

\section{Application to variant selection following gene-level testing}\label{sec-motivating-example}
 Genome-wide association studies (GWAS)  involves large scale association testing of genetic markers with underlying traits. 
 Large GWAS of uncommon and rare variants are now becoming increasingly feasible with the advent of newer genotyping chips, cheaper sequencing technologies and sophisticated algorithms that allow imputation of low-frequency variants based on combinations of common variants that are already genotyped in large GWAS. Thus,  association studies of rare variants is a very active area of research  and some of the early studies have already begun to report their findings, e.g., \cite {UK10K} and \cite{Fuchsberger16}.

As the statistical power for testing association of traits with individual rare variants may be low,  it has been suggested that tests for genetic associations be performed at an aggregated level by combining signals across multiple variants within  genomic regions such as those defined by functional units of genes \citep{Madsen09, Morris10, Neale11, Wu11, Lee12, Sun13}. These tests can be divided into sum-based tests (which aggregate the variant statistics by a linear combination), 
variance component tests (which aggregate the squared variant statistics by a linear combination), or combined (sum-based and variance component) tests. See \cite{Derkach14} for a review. There is, however, currently a lack of rigorous methods for variant selection following gene-level association testing.

The Dallas Heart Study (DHS) \citep{Romeo07} considered four genes of potential interest, genotyped in 3549 individuals (601 hispanic, 1830 non-hispanic black, 1043 non-hispanic white, 75 other ethnicities). We focus on the  32 variants in ANGPTL4,
which includes both rare and common variants. Table \ref{tab1}, column 2, shows the number of subjects with rare variants.

To detect associations with triglyceride (TG), a metabolism trait, we applied the variance component test SKAT of \cite{Wu11} %(with default weights)
, with outcome TG on a logarithmic scale, while adjusting for the covariates race, sex, and age on a logarithmic scale.  ANGPTL4 is one of the four genes in the ANGTPTL family \citep{Romeo09}. Using a Bonferroni correction for testing the genes in the family,  ANGTPL4 is selected for post-selection inference if the SKAT test $p$-value is at most $0.05/4$. To identify the potentially susceptible variants, we proceeded as suggested in section \S~\ref{sec-testing}.

The SKAT $p$-value for ANGTPL4 was $7.5\times 10^{-5}$ and therefore the gene was selected. Table \ref{tab1}, column 6 lists the weights assigned to each variant in the SKAT test. These weights were obtained using the default settings of the publicly available R library SKAT.
Figure \ref{fig1} and Table \ref{tab1} provide, respectively, a graphical display and the actual numbers for the naive (i.e., unconditional, not corrected for selection) and conditional $p$-values. When using the polyhedral method described in Section \S~\ref{subsec-poly}, one variant, $E40K$, passes the Bonferroni threshold for FWER control at the $0.05$ level. When using the hybrid method described in Section \S~\ref{subsec-hybrid}, two variants, $E40K$ and $R278Q$ are identified at an FDR level of $0.1$. This example demonstrates that it is possible to make further discoveries in a follow-up analysis after aggregate testing, to identify which underlying variants drive the signal. The variant $E40K$ is indeed associated with TG, as validated by external studies \citep{Dewey16}.  

\begin{figure}[t]
  \centering
     \includegraphics[width=14 cm]{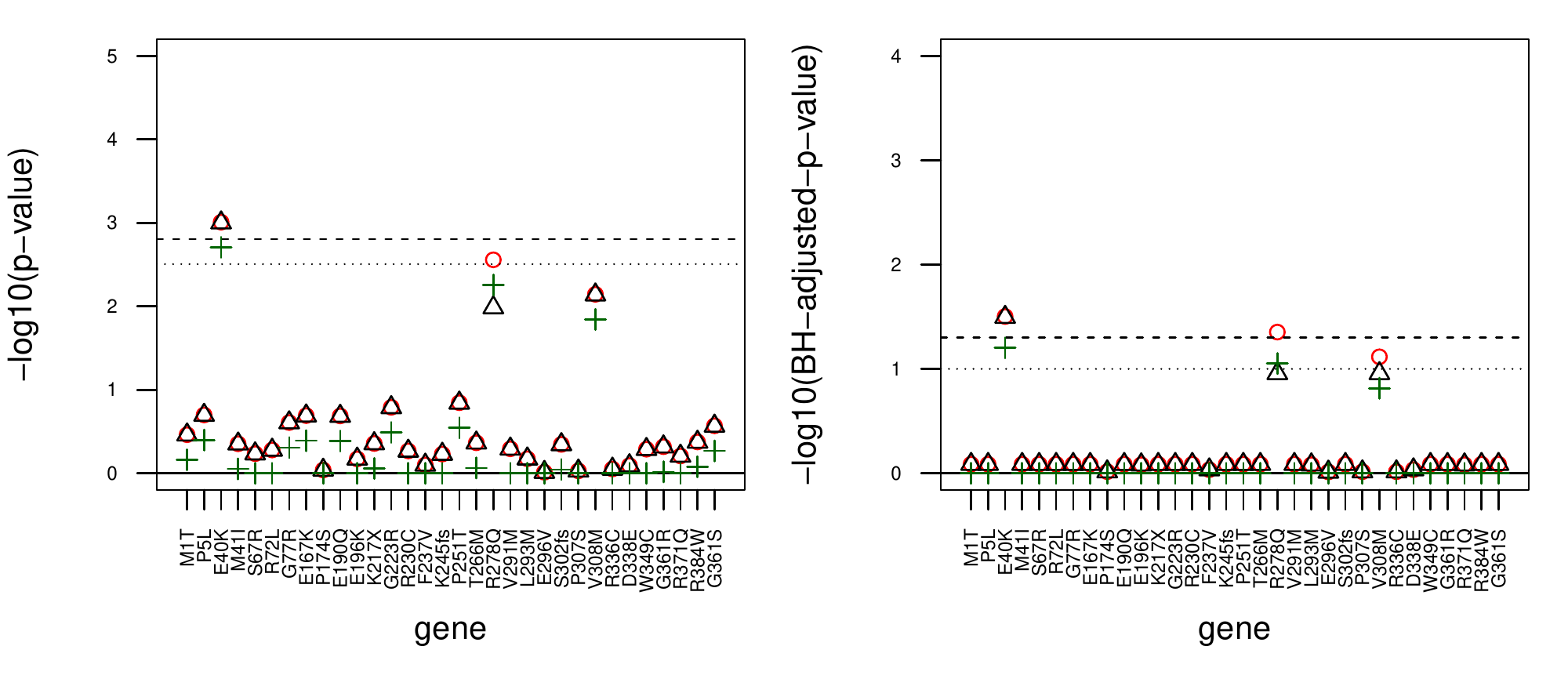} 
  \caption{The naive and two types selection adjusted $p$-values on a $-\log 10$ scale (left panel) and FDR adjusted p-values on a $-\log 10$ scale (right panel) for the 32 variants. The p-values plotted are Naive unadjusted p-values  (red circles), conditional p-values based on the polyhedral lemma (black triangles) and conditional p-values based on the hybrid method (green plus). The dotted line marks a multiplicity adjusted threshold of $0.1$ (FWER in the left panel, FDR in the right) and the dashed line marks a multiplicity adjusted threshold of $0.05$.}
\label{fig1}
\end{figure}

\section{Discussion}\label{sec-discuss}
In this work, we provided valid inference for linear contrasts of estimated parameters, after the aggregate test passed a pre-defined threshold.  
For the post-selection inference we suggest in this paper, we only need the summary statistics for the selected group of interest, and knowledge of the selection threshold $t_1$. The selection threshold does not have to be fixed. For example, a data dependent threshold will be valid if the groups are independent and via the BH procedure, or any other simple selection rule (as defined in 
 \citealp{Benjamini14}). If the data of all the groups is available, then there remains an open question of how to choose $t_1$ in order to maximize the chance of discovery for individual hypotheses (assuming that an error control guarantee at the group level is not necessary). Data adaptive methods for choosing $t_1$ may invalidate the post-selection inference. We are currently investigating potential approaches, but they are outside the scope of this manuscript.

Our methods can be extended  to tree structured hypothesis tests in a straightforward manner. See \cite{Bogomolov17} and the references within for state-of-the-art  work on hierarchical testing when there are more than two layers. An interesting genomic application is the following.  Within a selected gene, the tests may be further divided naturally into subgroups. For example, clusters of SNPs within a gene \citep{Yoo16}. It may be of interest to develop a multi-level analysis, where following selection we first examine the subgroups, and only then the individual effects. 

In this work we suggested switching regimes to adapt to the different unknown sparsity of the estimated effects. We observed that by combining a powerful method for  the sparse setting with a powerful method for the non-sparse setting, we get a method that has overall good performance.  Such an approach can be very useful in genomic applications, where the signal is expected to be sparse in some groups but non-sparse in others. 
The switching regime approach may benefit other post-selection settings as well, e.g., confidence intervals for the selected parameters in a  regression model.

\section{Supplementary material}
An R implementation of the methods in this paper is available in \url{https://github.com/ammeir2/PSAT} and will be available (soon) in the  Bioconductor package PSAT.  

\section*{Acknowledgement}
RH and AM contributed equally to this paper. The authors thank Andriy Derkach for providing his R implementation of the SKAT test, and for helpful conversations on the rare variant testing applications. Mathias Drton for helpful discussions regarding the properties of the conditional MLE. The authors are also grateful for Bin Zhu for helpful discussions of the Dallas Heart Study example. RH is supported by Israel Science Foundation grant no. 0603616831.

\newpage

\begin{center}
\begin{table}%[htbp]
%\begin{minipage}[b]{5\linewidth}
\caption{\label{tab1}\scriptsize{For the 32 varaints in ANGPTL4, the number of subjects with rare variants (column 2), the estimated effect size (column 3), the conditional $p$-value (column 4), the original $p$-value (column 5), the default Beta-density weight in SKAT (column 6), and the contribution of the variant to the SKAT statistic $\sum_{j=1}^{32} w_{m,j} U_j^2$ (column 7).  Variant $E40K$ has conditional $p$-value below the $0.05/32=0.0016$, and is therefore discovered by Bonferroni, with a guarantee of conditional FWER control at the 0.05 level. The contribution of variant R278Q is by far the largest towards the SKAT statistic, and therefore the conditional $p$-value is larger than the naive $p$-value for R278Q. For all other variants in this gene, the conditional $p$ values coincide with the naive $p$-values. This is expected when by conditioning on the test statistics of all the other variants, the SKAT test singificance is guaranteed regardless of the single variant test statistic value.
}} 
\scriptsize
%\centering
\fbox{
\begin{tabular}{*{8}{c}}
  \hline
Variant & \# rare variants & $\hat \beta$  & hybrid PV & conditional PV & naive PV & SKAT weight & $w_{m,j} U_j^2$ \\ 
  \hline
M1T & 1.0000 & 0.8967 & 0.6872 & 0.3434 & 0.3434 & 4.9831 & 19.9657 \\ 
  P5L & 2.0000 & 0.8588 & 0.3999 & 0.1995 & 0.1995 & 4.9663 & 74.7223 \\ 
  E40K & 50.0000 & -0.4490 & 0.0020 & 0.0010 & 0.0010 & 4.2172 & 7687.4667 \\ 
  M41I & 28.0000 & 0.1403 & 0.8860 & 0.4333 & 0.4333 & 4.5466 & 288.0531 \\ 
  S67R & 2.0000 & 0.3681 & 1.000 & 0.5827 & 0.5827 & 4.9663 & 15.9910 \\ 
  R72L & 3.0000 & -0.3468 & 1.000 & 0.5262 & 0.5262 & 4.9495 & 22.0274 \\ 
  G77R & 1.0000 & -1.0913 & 0.4928 & 0.2489 & 0.2489 & 4.9831 & 29.5739 \\ 
  E167K & 1.0000 & -1.2002 & 0.4072 & 0.2049 & 0.2049 & 4.9831 & 34.4125 \\ 
  P174S & 1.0000 & 0.1040 & 1.000 & 0.9125 & 0.9125 & 4.9831 & 0.4010 \\ 
  E190Q & 32.0000 & 0.2140 & 0.4108 & 0.2054 & 0.2054 & 4.4850 & 1199.4289 \\ 
  E196K & 1.0000 & -0.3991 & 1.000 & 0.6733 & 0.6733 & 4.9831 & 3.5122 \\ 
  K217X & 1.0000 & -0.7300 & 0.8704 &0.4406 & 0.4406 & 4.9831 & 12.4116 \\ 
  G223R & 1.0000 & -1.3242 & 0.3239 & 0.1621 & 0.1621 & 4.9831 & 40.5690 \\ 
  R230C & 1.0000 & -0.5784 & 1.000 & 0.5412 & 0.5412 & 4.9831 & 7.6586 \\ 
  F237V & 1.0000 & -0.2453 & 1.000 & 0.7956 & 0.7956 & 4.9831 & 1.2266 \\ 
  K245fs & 1.0000 & 0.5157 &  1.000 & 0.5858 & 0.5858 & 4.9831 & 6.6047 \\ 
  P251T & 1.0000 & 1.3860 & 0.2854 & 0.1432 & 0.1432 & 4.9831 & 49.3013 \\ 
  T266M & 1887.0000 & 0.0230 & 0.8619 & 0.2454 & 0.2454 & 0.0006 & 0.0002 \\ 
  R278Q & 207.0000 & -0.1945 & 0.0055 & 0.0103 & 0.0023 & 2.4309 & 10667.1021 \\ 
  V291M & 1.0000 & -0.6203 & 1.000 & 0.5123 & 0.5123 & 4.9831 & 9.5533 \\ 
  L293M & 1.0000 & 0.4021 & 1.000 & 0.6717 & 0.6717 & 4.9831 & 1.3213 \\ 
  E296V & 1.0000 & -0.0308 & 1.000 & 0.9741 & 0.9741 & 4.9831 & 0.0235 \\ 
  S302fs & 1.0000 & -0.7089 & 0.9012 & 0.4539 & 0.4539 & 4.9831 & 12.4794 \\ 
  P307S & 1.0000 & -0.0793 & 1.000 & 0.9333 & 0.9333 & 4.9831 & 0.0274 \\ 
  V308M & 3.0000 & 1.4719 & 0.0143 & 0.0071 & 0.0071 & 4.9495 & 477.6972 \\ 
  R336C & 7.0000 & -0.0469 & 1.000 & 0.8957 & 0.8957 & 4.8829 & 2.5696 \\ 
  D338E & 1.0000 & 0.2314 & 1.000 & 0.8073 & 0.8073 & 4.9831 & 0.0339 \\ 
  W349C & 1.0000 & -0.6120 & 1.000 & 0.5179 & 0.5179 & 4.9831 & 9.3008 \\ 
  G361R & 2.0000 & 0.4744 & 0.9598 & 0.4784 & 0.4784 & 4.9663 & 23.2973 \\ 
  R371Q & 1.0000 & 0.4724 & 1.000 & 0.6177 & 0.6177 & 4.9831 & 5.5410 \\ 
  R384W & 1.0000 & -0.7610 & 0.8420 & 0.4225 & 0.4225 & 4.9831 & 22.6686 \\ 
  G361S & 1.0000 & -1.0389 & 0.5388 & 0.2724 & 0.2724 & 4.9831 & 26.7995 \\ 
   \hline
\end{tabular}}
%\end{minipage}
\end{table}

\end{center}

\appendix

\section{Proof of Theorem \ref{thm1}}\label{app-thm1-proof}
In order to compute the distribution of a linear combination of $\hat {\vec \beta}$ within selected regions, it is useful to first represent the selection event in simple form (the representation is based on the one used for post-model selection in \citealp{Lee16}).

\begin{lemma}\label{lem1}
For an arbitrary linear combination ${\vec \eta}'{\hat {\vec \beta}}$, let  $c = ({\vec \eta}'{\bf \Sigma} {\vec \eta})^{-1}{\bf \Sigma} {\vec \eta}$ and 
${\vec W} = (I_{m}-c{\vec \eta}'){\hat {\vec \beta}}$, where $I_{m}$ is the $m\times m$ identity matrix. 
The selection event $S \geq S_{1-t_1}$  can be rewritten in terms of ${\vec \eta}'{\hat {\vec \beta}}$ as follows: 
$$\mathcal A(\vec W) = \left\{ {\vec \eta}'{\hat {\vec \beta}} \geq A({\vec W}), {\vec \eta}'{\hat {\vec \beta}} \leq B({\vec W})\right\},$$
where 
\begin{equation*}
\begin{array}{rl}
A({\vec W}) = \left\{
\begin{array}{rl}
\frac{-2{\vec W}'{\bf K}c+\sqrt{\Delta}}{2c'{\bf K}c} & \text{if } \Delta\geq 0,\\
-\infty & \text{if }\Delta < 0,\\
\end{array} \right. 
 & B({\vec W}) = \left\{
\begin{array}{rl}
\frac{-2{\vec W}'{\bf K}c-\sqrt{\Delta}}{2c'{\bf K}c} & \text{if } \Delta\geq 0,\\
\infty & \text{if }\Delta < 0,\\
\end{array} \right. 
\end{array} 
\end{equation*}
for $\Delta= 4({\vec W}'{\bf K}c)^2 - 4(c'{\bf K}c)({\vec W}'{\bf K}{\vec W}-S_{1-t_1})$. 
\end{lemma}
\begin{proof}
Decomposing ${\hat {\vec \beta}} = {\vec W}+c{\vec \eta}'{\hat {\vec \beta}}$, the result is immediate from rewriting $S=({\vec W}+c{\vec \eta}'{\hat {\vec \beta}})'{\bf K} ({\vec W}+c{\vec \eta}'{\hat {\vec \beta}})$  as a quadratic polynomial with argument ${\vec \eta}'{\hat {\vec \beta}}$. The 
 selection event is therefore $$\{ S\geq S_{1-t_1}\}   = \left\{[{\vec \eta}'{\hat {\vec \beta}}]^2c'{\bf K}c+ [{\vec \eta}'{\hat {\vec \beta}}]2{\vec W}'{\bf K}c + {\vec W}'{\bf K}{\vec W}-S_{1-t_1} >0\right\}.$$
\end{proof}

Since the covariance between ${\vec W}$ and ${\vec \eta}'{\hat {\vec \beta}}$ is zero, it follows from Lemma \ref{lem1} that if ${\vec \eta}'{\hat {\vec \beta}}$ is (approximately) normal, then the boundaries of the selection event are independent of ${\vec \eta}'{\hat {\vec \beta}}$. Therefore,  
conditional on ${\vec W}$ and on the selection event $\{S\geq S_{1-t_1} \}$,  ${\vec \eta}'{\hat {\vec \beta}}$ has a truncated normal distribution, truncated at the values $(-\infty, B({\vec W})]\cup [A({\vec W}), \infty)$. 

\section{Proof of Theorem \ref{thm-WALD-GN}}\label{app-thm-WALD-GN-proof} 
\proof
We shall show this without loss of generality for $j=1$, i.e.,  for testing $H_1: \beta_1=0$.  Since ${\bf K} = {\bf \Sigma}^{(-1)}$, it follows that $\vec W' K c = 0$. Therefore, the selection event is 
$$\{ S\geq S_{1-t_1}\}   = \left\{[{\vec \eta}'{\hat {\vec \beta}}]^2c'{\bf K}c+  {\vec W}'{\bf K}{\vec W}-S_{1-t_1} >0\right\}.$$ Clearly, the truncation will be smaller (i.e., $A(\vec W)$ smaller and $B(\vec W)$ larger), the larger ${\vec W}'{\bf K}{\vec W}$ is. It is clear from  the dependence of the distribution of  ${\vec W}'{\bf K}{\vec W}$ on $(0,\beta_2,\ldots,\beta_m)$ that it will be stochastically smallest for $(0,\beta_2,\ldots,\beta_m)=\vec 0. $ 
Therefore, 
\begin{eqnarray}
&&Pr_{(0,\beta_2,\ldots,\beta_m)'}(P_1\leq x \mid S>S_{1-t_1}) =  E\left[Pr_{(0,\beta_2,\ldots,\beta_m)'}(P_1\leq x \mid S>S_{1-t_1}, \vec W)\mid S>S_{1-t_1} \right] \nonumber \\
&& =  E_{(0,\beta_2,\ldots,\beta_m)}\left[1-F_{0, SE_1^2}^{\{\hat \beta_1\geq A(\vec W), \hat \beta_1\leq B(\vec W) \}}(x)\mid S>S_{1-t_1} \right]\nonumber \\ && \leq  E_{\vec 0}\left[1-F_{0, SE_1^2}^{\{\hat \beta_1\geq A(\vec W), \hat \beta_1\leq B(\vec W) \}}(x) \mid S>S_{1-t_1} \right] =  Pr_{\vec 0 }(P_1\leq x \mid S>S_{1-t_1}) = x, \nonumber
\end{eqnarray}
where the inequality follows since $1-F_{0, SE_1^2}^{\{\hat \beta_1\geq A(\vec W), \hat \beta_1\leq B(\vec W)\}}(x)$ is an increasing function of $A(\vec W)$ and a decreasing function of $B(\vec W)$, i.e., a decreasing function of $\vec W'{\bf K} \vec W$, so the expecation would be largest when $\vec W'{\bf K} \vec W$ is stochastically smallest, i.e., at $\vec \beta = \vec 0$. 
\qed
\section{Proof of Lemma \ref{lemma-gaussiancorrelationineq}}\label{app-lemma-gaussiancorrelationineq}
\proof
%Let $\hat{\vec\beta} \sim N(0, \Sigma)$, $K$ be a positive definite matrix and $S = \hat{\vec\beta}'K\hat{\vec\beta}$. 
For arbitrary fixed $b, s > 0$, define the following sets for some fixed index $j\in \{1,\ldots,m\}$
$$
A := \{\hat{\vec\beta}: \; S < s\}, \qquad
B := \{\hat{\vec\beta}: \hat\beta_j^{2} < b \}
$$
The sets $A$ and $B$ are both convex and symmetric about the origin if $\vec \beta = \vec 0$. By the Gaussian correlation inequality \citep{Royen14, Latala15} we have:
\begin{equation}\label{eq-AB}
Pr_{\vec \beta = \vec 0}(A,B) \geq Pr_{\vec \beta = \vec 0}(A)Pr_{\vec \beta = \vec 0}(B)
\end{equation}
The left-hand side of equation \eqref{eq-AB} can be written as
$$
Pr_{\vec \beta = \vec 0}(A,B) = 1 - Pr_{\vec \beta = \vec 0}(\hat\beta_j^{2} > b) - Pr_{\vec \beta = \vec 0}(S > s) + Pr_{\vec \beta = \vec 0}(\hat\beta_j^{2} > b, S > s),
$$
and similarly the right-hand side can be written as
$$
Pr_{\vec \beta = \vec 0}(A)Pr_{\vec \beta = \vec 0}(B) = 1 - Pr_{\vec \beta = \vec 0}(\hat\beta_j^{2} > b) - Pr_{\vec \beta = \vec 0}(S > s) + Pr_{\vec \beta = \vec 0}(\hat\beta_j^{2} > b) Pr_{\vec \beta = \vec 0}(S > s).
$$
Subtracting $1 - Pr_{\vec \beta = \vec 0}(\hat\beta_j^{2} > b) - Pr_{\vec \beta = \vec 0}(S > s)$ from both sides of \eqref{eq-AB} yields:
$$
Pr_{\vec \beta = \vec 0}(\hat\beta_j^{2} > b, S > s) \geq Pr_{\vec \beta = \vec 0}(\hat\beta_j^{2} > b) Pr_{\vec \beta = \vec 0}(S > s).
$$
Finally, 
$$
Pr_{\vec \beta = \vec 0}(\hat\beta_j^{2} > b | S > s) = \frac{Pr_{\vec \beta = \vec 0}(\hat\beta_j^{2} > b,S > s)}{Pr_{\vec \beta = \vec 0}(S > s)} 
\geq  \frac{Pr_{\vec \beta = \vec 0}(\hat\beta_j^{2} > b) Pr_{\vec \beta = \vec 0}(S > s)}{Pr_{\vec \beta = \vec 0}(S > s)} 
= Pr_{\vec \beta = \vec 0}(\hat\beta_j^{2} > b)
$$
%The left-hand side is $p'_{j,GN}$ and the right-hand side is $p_j'$ when $b$ and $s$ are the realized test statistics, and thus the result follows. 
\qed
\section{Proof of Theorem \ref{thm:linesearch}}\label{app-linesearch}
\proof
Let $\hat{\vec\beta}\sim N(\vec\beta,\bf\Sigma)$. Now, suppose that we are interested in solving the optimization problem:
$$
\max_{\vec\beta} \ell(\vec\beta) - \log Pr_{\vec\beta}(S >S_{1-t_1}).
$$
The above optimization problem can be rewritten as:
\begin{equation}\label{eq-gamma-split}
\max_{\gamma}\max_{\vec\beta \in B(\gamma)} \ell(\vec\beta) - \log \gamma,
\end{equation}
where
$$
B(\gamma) := \{\vec\beta: Pr_{\vec\beta}(S>S_{1-t_1}) = \gamma\}.
$$
In \eqref{eq-gamma-split} we divided our optimization problem into two parts. First, we must compute the maximizer of the likelihood for each power level $\gamma$ and then, we must maximize over $\gamma$ to find the global maximizer of the likelihood. The theorem hinges on the fact that this inner optimization problem has a closed form solution which we derive next. 

If $\bf K = \bf\Sigma^{(-1)}$ then the distribution of the test statistic is a non-central chi-square distribution, the parameters of which are the degrees of freedom (a known quantity) and the non-centrality parameter:
$$
\vec\beta'{\bf\Sigma^{(-1)}}\vec\beta.
$$
Thus, for each value of $\gamma \geq t_1$, there exists a $\delta \geq 0$ such that:
\begin{align} \nonumber
&\max_{\vec\beta\in B(\gamma)} \ell(\vec\beta)\\
=& \max_{\vec\beta} \ell(\vec\beta) , \qquad 
\text{s.t. } \vec\beta' {\bf\Sigma^{(-1)}} \vec\beta = \delta.
\label{eq-delta-constraint}
\end{align}
Now, for any $\delta \leq \hat{\vec\beta}'\bf\Sigma^{(-1)}\hat{\vec\beta}$ there exists a $c \geq 0$ such that the solution to \eqref{eq-delta-constraint} is given by:
$$
\max_{\vec\beta} \ell(\vec\beta) - c \vec\beta'{\bf\Sigma^{(-1)}}\vec\beta.
$$ 
This last problem, is a simple Tikhonov regularization problem, the solution which is given by $(1 + c)^{-1}\hat{\vec\beta}$ with:
$$
c = \sqrt{\frac{\hat{\vec\beta}'{\bf\Sigma^{(-1)}}\hat{\vec\beta}}{\delta}} - 1.
$$
Thus, for $\delta = 0$, $c = \infty$ and for $\delta = \hat{\vec\beta}'{\bf\Sigma^{(-1)}}\hat{\vec\beta}$, $c=0$ and we recover the least squares solution. From this, we can infer that $(1+c)^{-1}\in[0, 1]$

Because $(1 + c)^{-1} \in [0, 1]$ and all of the solution to the inner problem in \eqref{eq-gamma-split} are of the form $(1+c)^{-1}\hat{\vec\beta}$, we can conclude that the maximum likelihood estimator is given by:
$$
\tilde{\vec\beta} = \arg\max_{\lambda\in[0,1]}\ell(\lambda\hat{\vec\beta}) 
- \log Pr_{\lambda\hat{\vec\beta}}(S > S_{1-t_{1}}).
$$
\qed

\section{Proof of Theorem \ref{thm-ci}}\label{app-pf-thm-ci} 
Denote by $NC_j$ the event in which a non-covering confidence interval was constructed for $\beta_j$, by $NNC_j$ the event in which a naive confidence interval does not cover $\beta_j$ and by $CNC_j$ the event in which a conditional confidence did not cover $\beta_j$. 

In our procedure, if $S\geq S_{1-t_2}$, the confidence interval is based on the unconditional likelihood, and it is at level $1-\alpha$; if $S_{1-t_1}\leq   S < S_{1-t_2}$, the confidence interval is based on the exact conditional likelihood at $\vec \beta = \vec 0$, and it is at level  $1-(\alpha-t_2/t_1)$; otherwise, no confidence interval is constructed. Therefore, 
$$
Pr_{\vec\beta}(NC_j \mid S > S_{1 - t_1}) = 
$$
\begin{align*}
&= Pr_{\vec\beta}(NNC_j , S > S_{1- t_2}\mid S > S_{1 - t_1}) +
Pr_{\vec\beta}(CNC_j, S < S_{1 - t_2} \mid S > S_{1 - t_1})  \\
&\leq  Pr_{\vec\beta}(S > S_{1 - t_2} \mid S > S_{1 - t_1}) + 
Pr_{\vec\beta}(CNC_j \mid S > S_{1 - t_1}).
\end{align*}
If $\vec \beta = \vec 0$, then $Pr_{\vec 0}(S\geq S_{1-t_2}\mid S>S_{1-t_1}) = t_2/t_1$ and $Pr_{\vec 0}(CNC_j\mid S< S_{1-t_2}) = \alpha - t_2/t_1$. Therefore, 
$Pr_{\vec 0}(NC_j\mid S>S_{1-t_1})\leq \alpha$. 

If $\vec \beta \neq  \vec 0$, then $\lim_{n\rightarrow \infty} Pr_{\vec \beta}(S>S_{1-t_2}) = 1$. This follows for the linear model,  since $E(\hat {\vec \beta} - \vec \beta) = \vec 0$ and $var\left(\hat {\vec \beta}\right) = ({\bf X}'{\bf X} )^{(-1)}var(\epsilon_1)$. This also follows for the logistic model, since 
large $n$, $E(\hat {\vec \beta} - \vec \beta) = O(\frac 1n)$ and $var\left(\hat {\vec \beta}\right) = ({\bf X}' {\bf W} {\bf X})^{(-1)}(1+O(\frac 1n))$.  Therefore, 
\begin{eqnarray}
\lim_{n\rightarrow\infty} Pr_{\vec \beta}(NC_j\mid S>S_{1-t_1})=\lim_{n\rightarrow \infty} Pr_{\vec \beta}(NNC_j\mid S\geq S_{1-t_2}) = \alpha \nonumber. 
\end{eqnarray}
\qed

\section{MLE computation for generalized linear models}\label{appendix-mle-computation}
Suppose that we observe independent draws $(y_1,\vec X_1),...,(y_n,\vec X_n)$ where $y_i | \vec X_i \sim f_{\vec\beta}$ follows an exponential family distribution with $E(y_i) = g^{(-1)}(\vec X_i \vec\beta)$ for some link function $g$. If $n$ is small or ${\bf X}$ is sparse, it may be undesirable to assume that $\hat{\vec\beta}$ has a normal distribution. In such a case, we can obtain an exact MLE using a stochastic gradient method while sampling the gradient steps from the post-selection distribution of $\vec{y}$ given ${\bf X}$. We start by describing one possible algorithm for sampling from the desired post-selection distribution, and then briefly discuss the conditions for convergence of the stochastic gradient method. 

If $t_1$ is not too small or $\vec\beta$ is sufficiently large, one can obtain samples from the post-selection distribution of $\vec{y} | S > S_{1- t_1}$ using a rejection sampler. Otherwise, it possible to use the following, general purpose Metropolis-Hastings algorithm. Fix a parameter value $\vec\beta$, initialize $\vec y_0 = y$ and set some (preferably large) integer $J > 0$. Then, repeat for $j \in \{1,\dots,J\}$ and $i\in\{1,\dots,n\}$:
\begin{enumerate}
\item Sample $y_{ij} \sim f_{\vec\beta}(y_{i} \mid  \vec X_i)$. 
\item Compute the test statistic $S_{ij}$ as determined by the current state of the chain.
\item If $S_{ij} < S_{1 - t_1}$ then set $y_{ij} \leftarrow y_{i(j-1)}$.
\end{enumerate}
The algorithm defined by these steps is guaranteed to converge to the correct post-selection distribution of $\vec{y}|S > S_{1 - t_1}$. 

Given samples form the post-selection distribution of $\vec{y}(\vec\beta)$ taken at a specific parameter value $\vec\beta$ we can take stochastic gradient steps of the form:
$$
\tilde{\vec\beta}^{t + 1} = \tilde{\vec\beta}^{t} + T(\vec y) - T(\vec y(\tilde{\vec\beta}^{t}))
$$
where $T(\vec y)$ is the observed sufficient statistic for $\vec\beta$ and $T(\vec y(\tilde{\vec\beta}^{t}))$ is the sampled sufficient statistic. In order to guarantee the convergence of the stochastic gradient steps to the correct MLE, one must verify that there exists a constant $A  > 0$ such that for all $\vec\beta$ and $j \in \{1,...,m\}$:
\begin{equation}\label{eq-w-bound}
E_{\vec\beta}\| T_j(\vec y) - E_{\vec\beta}[T_j(\vec y)|S > S_{1 - t_1}] \|^{2}
 \leq 
 A \left(1 + \|T(\vec y) - E_{\vec\beta}[T(\vec y |S > S_{1-t_1}]\|^{2}\right).
\end{equation}
In \eqref{eq-w-bound}, $T(\vec y)$ on the right-hand side of the inequality is the observed sufficient statistic. This condition holds for example, for logistic regression and linear regression with normal errors. See \citet{Bertsekas00} and \citet{Meir17} for details.

\section{Theory and methods for inference after screening with linear tests}\label{sec-linear-aggregate-test}
Suppose that we observe $\hat{\vec\beta}\sim N_m(\vec\beta, {\bf \Sigma})$ and estimate $\vec\beta$ if and only if:
$$
\vec a' \hat{\vec\beta} < l \qquad \text{or} \qquad \vec a'\hat{\vec \beta} > u,
$$
where $\vec a \in \Re^{m}$ and $l < u$ are some pre-specified constants. One common choice is to set $\vec a =  (1,\dots,1)$ and 
$$
l = z_{t_1/2} \sqrt{\vec a' {\bf \Sigma}\vec a}, \qquad 
u = z_{1-t_1/2} \sqrt{\vec a' {\bf \Sigma}\vec a},
$$
to obtain a level $t_1$ aggregate test. In this section we will develop inference methods equivalent to the ones presented in the main body of the paper for inference after screening with linear aggregate tests. All of the theoretical results in this section are simple generalization or straightforward corollaries of the theory developed for inference after quadratic tests. 

In \S~\ref{sub-linear-poly} we state the Polyhedral Lemma for inference after selection with an aggregate test. In \S~\ref{sub-linear-hybrid} we define the hybrid p-values for testing after screening with a linear aggregate test and provide theoretical guarantees for their validity. The theoretical guarantee relies on the result  in \S~\ref{thm-cons-proof},  which identifies the parameters for which the conditional $p$-value is most conservative following selection by a linear test. In \S~\ref{sub-linear-conditional} we develop a formula for computing the conditional MLE via a line search. In \S~\ref{sub-linear-ci} we define the regime switching confidence interval for the linear aggregate testing case. Finally, \S~\ref{sub-linear-simulation} is a short simulation study. 

\subsection{The polyhedral Lemma for selection with a linear test}\label{sub-linear-poly}
In this section we develop the polyhedral Lemma for the inference after a linear aggregate testing problem. Let $\vec\eta \in \Re^{m}$ and suppose that we are interested in inferring on $\vec\eta' \vec\beta$. As in \S~\ref{app-thm1-proof} we write:
$$
\hat{\vec\beta} = \vec c\vec\eta'\hat{\vec\beta} + \vec W.
$$
Denote by $s_{ac}$ the sign of $\vec{a}'\vec {c}$. Recall that our selection criterion was
$$
\mathcal{S} := \left\{\vec a' \hat{\vec\beta} < l \cup \vec a'\hat{\vec\beta} > u\right\}
$$
and so the truncation for $\vec\eta'\hat{\vec\beta}$ conditional on $\vec W$ is given by
$$
s_{ac} \vec\eta'\hat{\vec\beta} < s_{ac} \frac{l - \vec a'\vec W}{\vec a'\vec c} = A(\vec W) , \qquad
s_{ac} \vec\eta'\hat{\vec\beta} > s_{ac} \frac{u - \vec a'\vec W}{\vec a' \vec c} = B(\vec W).
$$

\begin{theorem}
Let $\vec\eta'\vec{\hat\beta}$ be a linear combination of $\hat{\vec\beta}$, $l < u$ fixed selection thresholds and let $\vec W$ be as defined in Theorem \ref{thm1}. Then:
$$
\vec\eta' \hat{\vec\beta} | \mathcal{S}, \vec W \sim TN(\vec\eta \vec\beta, \vec\eta'{\bf\Sigma}\vec\eta, \mathcal{A}(\vec W))
$$
with
$$
\mathcal{A}(\vec W) = \left\{s_{ac} \vec\eta'\hat{\vec\beta} <  B(\vec W) 
\cup 
s_{ac} \vec\eta'\hat{\vec\beta} >  A(\vec W)
\right\}.
$$
\end{theorem}

\subsection{Hybrid p-values after linear aggregate testing}\label{sub-linear-hybrid}
As in the case of inference after aggregate testing with a quadratic test, we propose to combine p-values computed using the polyhedral lemma and p-values computed under the global-null in order to obtain a powerful post-selection test. In this section we address two types of aggregate tests, the symmetric two-sided aggregate test with $l = -u$, and the one-sided aggregate test with $l = -\infty, u = u$. We discuss the more general case of $-\infty <l \leq u <\infty$ in Section \ref{thm-cons-proof}.

Suppose that we are interested in testing the null hypothesis $H_{0\eta}: \vec\eta'\vec\beta =0$. Then, the event under the null the value of $\vec\beta$ influences the computation of the post-selection p-value. To see why, consider for example a case where we specify a value for $\vec\beta$ such that $E(\vec a' \vec W)$ is very large. Under such a parametrization we will have $P(\mathcal{S}) \approx 1$ and selection adjusted p-value for testing $H_{0\eta}$ will be identical to the unadjusted p-value. On the other hand, if we set $\vec\beta = \vec 0$ then the adjusted p-value will always be larger than the unadjusted one. In fact, for a two sided aggregate test the choice of $\vec\beta = \vec 0$ will yield the most conservative test possible.   
\begin{corollary}\label{the-corollary}
Suppose $H_{0\eta}$ was selected for inference via a linear aggregate test with $-\infty < -u <0 < u < \infty$ and that $\vec \beta$ satisfies $H_{0\eta}$. Define the two-sided post-selection p-value:  
$$
p_{\vec\beta}(b) = 2 \min\left(
Pr_{\vec\beta}(\eta'\hat{\vec\beta} \geq b|\mathcal{S}),
Pr_{\vec\beta}(\eta'\hat{\vec\beta}) \leq b| \mathcal{S})
\right).
$$
Then, 
\begin{equation*}
Pr_{\vec\beta}\left(p_{\vec 0}(\vec\eta'\hat{\vec\beta}) < t \middle |\mathcal{S}\right)
\leq
Pr_{\vec 0}\left(p_{\vec 0}(\vec\eta'\hat{\vec\beta})< t\middle| \mathcal{S}\right) 
=t, \qquad \forall t\in(0, 1).
\end{equation*}
\end{corollary}
Corollary \ref{the-corollary} is a special case of a more general result described in Section \ref{thm-cons-proof}

In the case of screening with a one-sided test, there does not exist a most conservative p-value. However, the choice of $\vec\beta =\vec 0$ is an asymptotically conservative one. To see why, consider three case. The first and easiest case is $\vec\beta = \vec 0$, here our test will be an exact one. If $\vec\beta$ is such that $E(\vec a' \vec W) < 0$ then the probability of selection decreases to zero at an exponential rate under the null, making this case irrelevant from an asymptotic point of view. Finally if the true $\vec\beta$ is such that $E(\vec a' \vec W) > 0$ then selection occurs with probability approaching $1$ regardless of the value of $\vec\eta'\vec\beta$ and so any selection adjustment is asymptotically conservative. 

To summarize, we define the hybrid p-value for testing $H_{0j}:\beta_j = 0$. 
\begin{enumerate}
\item Compute the a p-value under the global-null $p'_{j,GN}$. 
\item Compute an adjusted p-value based on the polyhedral lemma $p'_{j}$.
\item Set $p'_{j,hybrid} = 2\min(p'_j,p'_{j,GN})$.
\end{enumerate}
Because both the polyhedral and global-null p-values are (asymptotically) valid, the hybrid p-value is also a valid p-value. 

\begin{remark} Sampling from the conditional distribution. \emph{
Sampling following selection via a linear aggregate test is far simpler than sampling following selection using a quadratic test. Sampling can be done in two steps, the first involves sampling from the univariate truncated normal distribution
$$
\bar{b} \sim f_{\vec\beta}\left(\vec a' \hat{\vec\beta} |\mathcal{S}\right)
$$
and the second is sampling from the multivariate normal distribution conditional on the contrast
$$
\vec b \sim 
f_{\vec\beta}\left(\hat{\vec\beta}|\mathcal{S}, \vec a' \hat{\vec\beta} = \bar{b}\right) = 
f_{\vec\beta}\left(\hat{\vec\beta}|\vec a' \hat{\vec\beta} = \bar{b}\right).
$$
}
\end{remark}

\subsection{Conditional estimation after linear aggregate testing}\label{sub-linear-conditional}
In this section we discuss the computation of the conditional-MLE for a model which was selected via a linear aggregate test. The distribution of $\vec a' \hat{\vec\beta}$ is determined by a single unknown parameter $\vec a ' \vec\beta$ and it is therefore possible to cast the problem of computing the multivariate conditional MLE as a line-search problem. Let $\delta := \vec a'\vec\beta$. The conditional MLE for the linear aggregate testing problem is
$$
\tilde{\vec\beta} = \arg\max_{\vec\beta} \ell(\vec\beta) - \log Pr_{\vec\beta} (\mathcal{S}).
$$
Once again we re-write the optimization problem as a nested optimization problem:
\begin{align*}
&\max_{\vec\beta} \ell(\vec\beta) - \log Pr_{\vec\beta} (\mathcal{S}) \\
=& \max_{\delta}\max_{\vec\beta: a' \vec\beta = \delta} \ell(\vec\beta) -\log Pr_{\delta} (\mathcal{S}).
\end{align*}
For a fixed $\delta$, the solution to the inner problem is given by:
$$
\tilde\beta(\delta) = \hat{\vec\beta}  + \frac{\delta - \vec a' \hat{\vec\beta}}{\vec a' {\bf\Sigma} \vec a} {\bf \Sigma} \vec a,
$$
and the MLE is given by $\tilde{\vec\beta}(\tilde{\delta})$ with:
$$
\tilde{\delta} = \arg\max_{\delta} \ell(\tilde{\vec\beta}(\delta)) - \log Pr_{\delta}(\mathcal{S}).
$$

\subsection{Confidence intervals after linear aggregate testing}\label{sub-linear-ci}
For inference after selection with linear aggregate tests we construct confidence intervals using Procedure \ref{proc-postselectionCI}. The use of regime switching confidence intervals allows us to safely assume that $E(\vec a'\vec W) = 0$ as asymptotic validity is provided by the fact that if $\vec a' \vec \beta \neq 0$ then selection occurs with probability approaching one triggering the switch in regimes and if $\vec a' \vec\beta = 0$ then test-inversion confidence intervals based on the global-null assumption are exact. 

\begin{figure}[t]
\label{fig-est-ci-linear}
\begin{center}
\includegraphics[width= 14 cm]{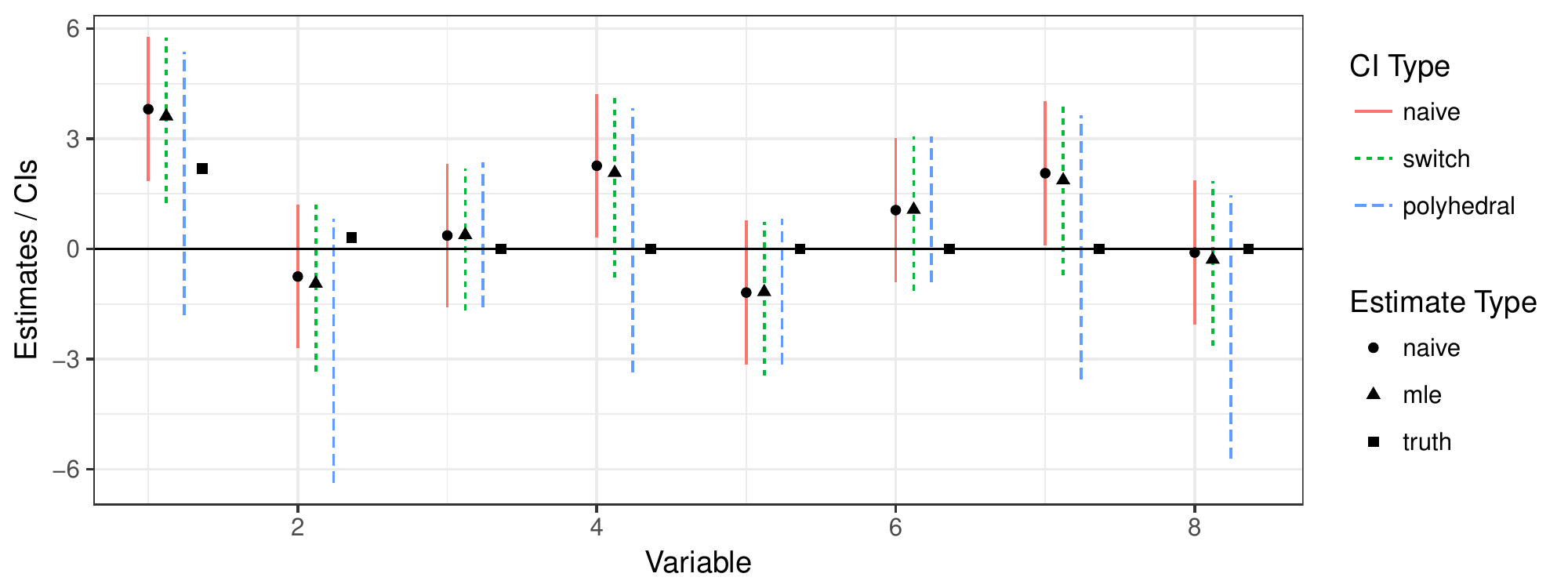}
\caption{Point estimates and confidence intervals for the artificial data example described in Example \ref{linear-ci-example}. The naive point estimate as marked as a circle, the conditional MLE as a triangle and the true value of the parameter is marked as a square. The confidence intervals are the naive (solid red line), hybrid (dotted green line) and polyhedral (dashed blue line).}
\end{center}
\end{figure}

\begin{example}\label{linear-ci-example}\emph{
In order to demonstrate the proposed post-selection inference methods, we conduct an experiment with an artificial dataset. We generate $X$ and $\hat\beta$ as in \S~\ref{sec-sim} setting $\vec \beta = (2.2, 0.3, 0, \dots, 0)$. We base our aggregate test on the contrast $\vec a = (1, 1, -1, 1, -1, -1, 1, 1)$ and set our lower and upper thresholds in such a way as to obtain a level $t_1 = 0.01$ test. The results of the analysis for the artificial dataset are presented in Figure \ref{fig-est-ci-linear}. Notice that in the linear case the conditional MLE does not shrink all of the variables, but may inflate some and shrink others. In inference after aggregate testing the polyhedral confidence intervals have a tendency to be much larger than ones computed under the global-null hypothesis, meaning that the regime switching confidence intervals tend to be more efficient. 
}
\end{example}

\begin{figure}[t]
\begin{center}
\includegraphics[width= 14 cm]{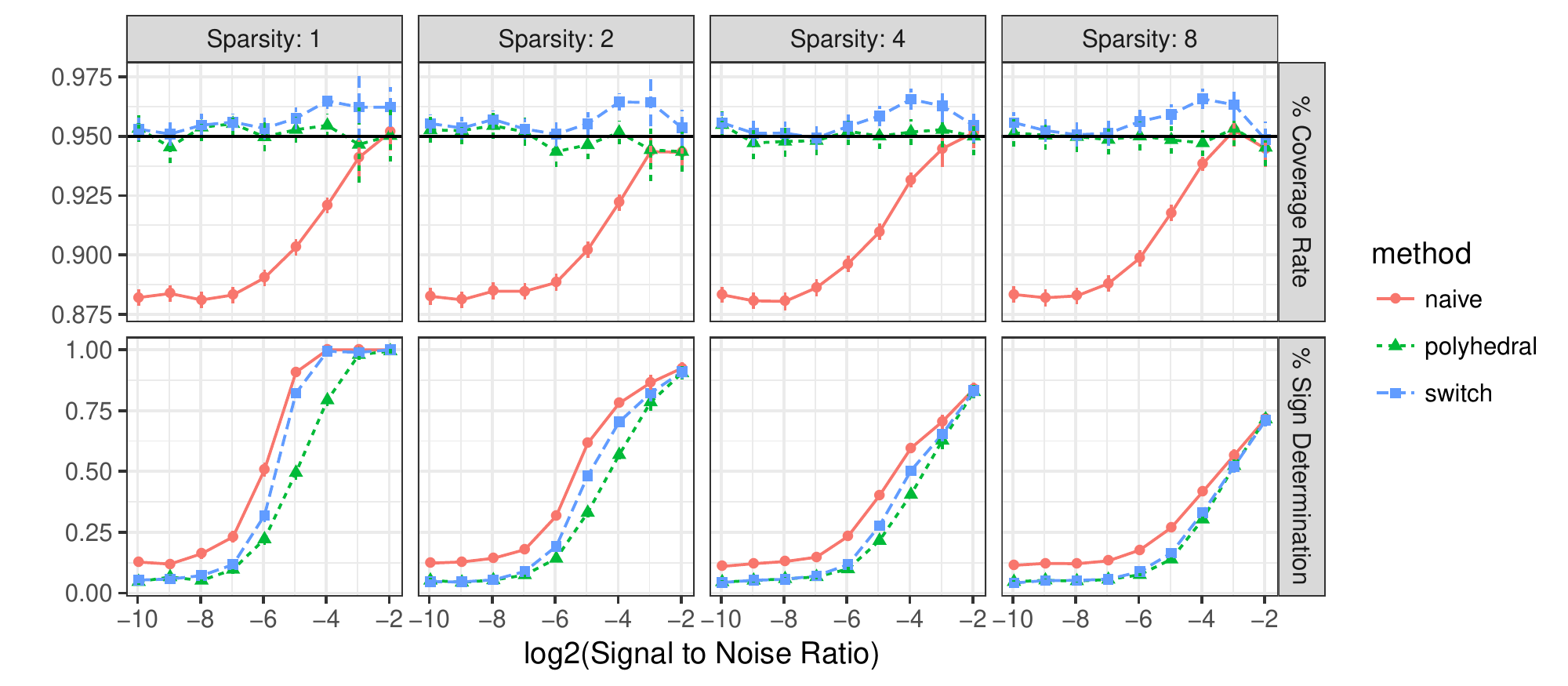}
\caption{Coverage rates and power to determine the sign of confidence intervals constructed after aggregate testing. We plot results rates for the naive unadjusted confidence intervals (solid red line), polyhedral confidence intervals (dotted green line) and regime switching intervals with $t_2 = t_1\alpha^2$ (dashed blue line).}\label{fig-cover-sim-linear}
\end{center}
\end{figure}

\subsection{Simulations for inference after linear aggregate testing}\label{sub-linear-simulation}
To conclude this section we conduct a simulation study with the goal of verifying our inference methods. We generate data as in \S~\ref{sub-cover}. We keep datasets that pass a symmetric linear aggregate test at a $t_1 = 0.001$ level using a contrast $\vec a$, the coordinates of which are equal to $\text{sign}(\beta_j)$ for $\beta_j \neq 0$ and are set to $1$ or $-1$ at random for $\beta_j = 0$. The results are presented in Figure \ref{fig-cover-sim-linear}. As in the case of screening with the Wald test, our proposed inference methods achieve the desired coverage rates while the naive confidence intervals tend to have a lower than nominal coverage rate when the signal to noise ratio is low. As could be expected, the regime switching confidence intervals have better power than the polyhedral confidence intervals to determine the sign, especially when the true model is sparse. 

\subsection{Asymmetric aggregate tests}\label{thm-cons-proof}
In this section we discuss the problem of computing a conservative p-value when $-\infty <l <u<\infty$ and $l\neq -u$. While screening with an asymmetric two-sided case is not common in practice, the analysis of this problem yields an interesting and somewhat surprising result. Specifically, we find that the most conservative parametrization of $\vec\beta$ for testing $H_{0\eta}:\vec\eta'\vec\beta$ is such that $\vec\eta' \vec\beta = 0$ and
\begin{equation}\label{l-u-condition}
E(\vec a' \vec W) = \frac{l + u}{2}.
\end{equation}
This parametrization can be used to compute hybrid p-values after testing with an asymmetric two-sided aggregate test. However, we note that if the test is highly imbalanced (e.g. $l = -10, u =1$) then the most conservative parametrization will yield a very conservative test and it might be preferable to use the asymptotically valid global-null p-value. 

\begin{theorem}\label{thm-linear-conservative}
Let $\hat{\vec\beta} \sim N(\vec\beta, {\bf \Sigma})$ and suppose that $a\neq \eta$, $\vec\eta' \vec\beta = 0$ and that $-\infty < l, u < \infty$ and let:
$$
p_{\vec\beta}(b) = 2 \min\left(
Pr_{\vec\beta}(\eta'\hat{\vec\beta} \geq b|\mathcal{S}),
Pr_{\vec\beta}(\eta'\hat{\vec\beta}) \leq b| \mathcal{S})
\right).
$$
Then, for $\tilde{\vec\beta}$ that also satisfies \eqref{l-u-condition} we have
\begin{equation}\label{pval-inequality}
Pr_{\vec\beta}\left(p_{\tilde{\vec\beta}}(\vec\eta'\hat{\vec\beta}) < t \middle |\mathcal{S}\right)
\leq
Pr_{\tilde{\vec\beta}}\left(p_{\tilde{\vec\beta}}(\vec\eta'\hat{\vec\beta})< t\middle| \mathcal{S}\right) 
=t, \qquad \forall t\in(0, 1).
\end{equation}
If $a = \eta$ then any parametrization of $\vec\beta$ that satisfies $\vec\eta'\vec\beta=0$ yields a valid p-value.
\end{theorem}

Before we get into the technicalities of the proof of Theorem \ref{thm-linear-conservative}, let us break down the component of equation \eqref{pval-inequality}. We have two quantities that depend on the parameter values under which we evaluate the p-value. The first is the p-value itself $p_{\vec\beta}(b)$ which effectively determines the threshold for declaring that a test is rejected at a level $t$ and it is evaluated under the same parameter value on both sides of the inequality. The second component is the parameter value under which we evaluate the probability of crossing a threshold $Pr_{\vec\beta}(p_{\vec\beta}({\vec\eta}'\hat{\vec\beta}) < t)$ which determines under what set of parameters we evaluate the probability of crossing the thresholds determined by $p_{\vec\beta}$. If we can show that equation \eqref{pval-inequality} holds then this implies that evaluating the probability of crossing the threshold at $\tilde{\vec\beta}$ is the most conservative, making our procedure conservative.

Assume w.l.o.g that $a'c = 1$. We begin by noting that the joint distribution of $w := \vec a \vec W$ and $\eta\hat{\vec\beta}$ is that of independent normal vector, and so, examining the marginal density of $w$ under the null and the assumption that $E(w) = (l + u)/2$, we can see that $w$ has a symmetric distribution about $(l + u) /2$:
\begin{align*}
f(w|\mathcal{S}) &= 
\frac{P(\mathcal{S}|w)}{P(\mathcal{S})} \varphi(w) \\ &=
\frac{P(\{\vec\eta' \hat{\vec{\beta}} < l - w\}\cup \{ \vec\eta' \hat{\vec\beta} >  u - w\})}{P(\mathcal{S})} \varphi(w) \\ &=
\frac{P(\{\vec\eta' \hat{\vec\beta} < \frac{l-u}{2} - (w - \frac{l+u}{2})\}\cup \{ \vec\eta' \hat{\vec\beta} >  \frac{u-l}{2} - (w - \frac{l+u}{2})\})}{P(\mathcal{S})}
\varphi(w)
\end{align*}
and so, the truncation has a symmetric distribution about $0$ in the sense that
$$
A(\vec W) = \frac{l-u}{2} - \left(w - \frac{l+u}{2}\right)=^D   
-\left(\frac{u - l}{2} - \left(w - \frac{l+u}{2}\right)\right)=
-B(\vec W).
$$
Thus, there exists a constant $b(t)$ such that:
\begin{align*}
Pr_{\tilde{\vec\beta}}(p_{\tilde{\vec\beta}} < t | \mathcal{S}) = 
Pr_{\tilde{\vec\beta}}(\vec\eta'\hat{\vec\beta} > b(t)| \mathcal{S}) + 
Pr_{\tilde{\vec\beta}}(\vec\eta'\hat{\vec\beta} < -b(t)| \mathcal{S}).
\end{align*}

Fixing a value for $w$, we have
\begin{align}\label{eq-cond-on-w}
Pr_{H_0}(p_{\tilde{\vec\beta}}(\vec\eta'\hat{\vec\beta} )< t| \mathcal{S}, \vec a' \vec W = w)
&=
\frac{Pr(\vec\eta'\hat{\vec\beta} < -b(t), \mathcal{S}(w)) + Pr(\vec\eta'\hat{\vec\beta} > b(t),\mathcal{S}(w))}
{P(\mathcal{S}(w))} \\
\mathcal{S}(w) &:= \{\vec\eta'\hat{\vec\beta} < l - w\}\cup\{\vec\eta'\hat{\vec\beta} > u - w\} 
\nonumber
\end{align}
and 
$$
Pr_{\vec\beta}(p_{\tilde{\vec\beta}}(\vec\eta'\hat{\vec\beta} )< t| \mathcal{S}) = 
\int_\mathcal{\Re} Pr_{H_0}(p_{\tilde{\vec\beta}}(\vec\eta'\hat{\vec\beta} )< t| \mathcal{S}, \vec a' \vec W = w) f_{\vec\beta}(w|\mathcal{S})dw.
$$
Notice that \eqref{eq-cond-on-w} is symmetric about $w = (l + u) / 2$. 

Taking a derivative, 
\begin{align}\label{eq-w-deriv}
\frac{\partial}{\partial E(w)} Pr_{\vec\beta}(p_{\tilde{\vec\beta}}(\vec\eta'\hat{\vec\beta} )< t| \mathcal{S}) &= \\
&\int_{\Re}Pr_{H_0}(p_{\tilde{\vec\beta}}(\vec\eta'\hat{\vec\beta} )< t| \mathcal{S}, \vec a' \vec W = w)  \frac{\partial}{\partial E(w)}f(w|\mathcal{S}) dw. \nonumber
\end{align}
The inner derivative equals:
$$
\frac{\partial}{\partial E(w)} f(w|\mathcal{S}) = 
P(\mathcal{S}|w) \varphi(w)
\frac{w - E(w|\mathcal{S})}
{\sigma^{2}_wP(\mathcal{S})}.
$$
The derivative in \eqref{eq-w-deriv} equals zero at $E(w) = (l + u)/2$ because for such a parameter value $E(w|\mathcal{S}) = (l+u)/2$, $f'(w - (l + u)/2|S) = -f'((l + u)/2 - w|S)$ and  
$$
Pr_{H_0}(p_{\tilde{\vec\beta}}(\vec\eta'\hat{\vec\beta} )< t| \mathcal{S}, \vec a' \vec W = w - (l + u) / 2) = Pr_{H_0}(p_{\tilde{\vec\beta}}(\vec\eta'\hat{\vec\beta} )< t| \mathcal{S}, \vec a' \vec W = (l + u)/2 - w),
$$
this is also the only maximum because the inner derivative is symmetric only at $E(w) = (l + u)/ 2$.

Finally, if $\vec a=\vec\eta$ then $\vec a' \vec W=0$ by definition and therefore the distribution of $\vec\eta'\hat{\vec\beta}$ is always truncated normal constrained to $(-\infty, l]\cup[u,\infty)$.
\qed

 \end{document}